\documentclass[times,twocolumn,final]{elsarticle}
\usepackage{medima}
\usepackage{marginnote}
\usepackage{tabularx}
\usepackage{booktabs}
\usepackage{amssymb}
\usepackage{amsthm}
\usepackage{bm}
\usepackage{bbm}
\usepackage{subfig}
\usepackage{graphicx}

\usepackage{adjustbox}

\usepackage{amsmath}

\DeclareMathOperator*{\argmin}{arg\,min}

\theoremstyle{plain}

\newtheorem{definition}{Definition}[section]

\newtheorem{lemma}{Lemma}[section]

\newcommand\norm[1]{\left\lVert#1\right\rVert}

\usepackage{xspace}
\newcommand\Tone{T\textsubscript{1}\xspace}
\newcommand\Tonec{T\textsubscript{1}c\xspace}
\newcommand\Ttwo{T\textsubscript{2}\xspace}

\newcolumntype{Y}{>{\centering\arraybackslash}X}
\newcolumntype{C}[1]{>{\centering\let\newline\\\arraybackslash\hspace{0pt}}m{#1}}
\newtheorem{prop}{Proposition}
\usepackage{hyperref}

\journal{Manuscript accepted in Medical Image Analysis}
\begin{document}

\verso{Dorent \emph{et~al.}}

\begin{frontmatter}

\title{Learning joint segmentation of tissues and brain lesions from task-specific hetero-modal domain-shifted datasets}%

\author[1]{Reuben \snm{Dorent}}
\ead{reuben.dorent@kcl.ac.uk}
\author[1,2]{Thomas \snm{Booth}}
\author[1,5]{Wenqi \snm{Li}}
\author[1,3,4]{Carole H. \snm{Sudre}}
\author[2]{Sina \snm{Kafiabadi}}
\author[1]{Jorge \snm{Cardoso}}
\author[1]{Sebastien \snm{Ourselin}}
\author[1]{Tom \snm{Vercauteren}}
%\author[]{for the Alzheimer's Disease Neuroimaging Initiative}
%\corref{cor1}

\address[1]{King's College London, School of Biomedical Engineering \& Imaging Sciences, St. Thomas' Hospital, London. United Kingdom.}
\address[2]{Department of Neuroradiology, King's College Hospital NHS Foundation Trust, London. United Kingdom.}
\address[3]{Dementia Research Centre, UCL Institute of Neurology, UCL, London. United Kingdom.}
\address[4]{Department of Medical Physics, UCL, London. United Kingdom.}
\address[5]{NVIDIA, Cambridge. United Kingdom.}

%\cortext[cor1]{Some of the data used in preparation of this article was obtained from the Alzheimer’s Disease Neuroimaging Initiative (ADNI)  database  (\url{adni.loni.usc.edu}).  As  such,  the  investigators  within  the  ADNI  contributed  to  the  design and implementation of ADNI and/or provided data but did not participate in analysis or writing of this report. A complete listing of ADNI investigators can be found at:\url{http://adni.loni.usc.edu/wp-content/uploads/how_to_apply/ADNI_Acknowledgement_List.pdf}}

% \received{1 May 2013}
% \finalform{10 May 2013}
% \accepted{13 May 2013}
% \availableonline{15 May 2013}
% \communicated{S. Sarkar}

\begin{abstract}
%%%
Brain tissue segmentation
from multimodal MRI
is a key building block of many neuroimaging analysis pipelines.
Established tissue segmentation approaches have, however, not been developed to cope with large anatomical changes resulting from pathology, such as white matter lesions or tumours, and often fail in these cases.
In the meantime,
with the advent of deep neural networks (DNNs), segmentation of brain lesions has matured significantly.
However, few existing approaches allow for the joint segmentation of normal tissue and brain lesions. 
Developing a DNN for such a joint task is currently hampered by the fact that annotated datasets typically address only one specific task and rely on task-specific imaging protocols including a task-specific set of imaging modalities.
In this work, we propose a novel approach to build a joint tissue and lesion segmentation model from aggregated task-specific hetero-modal domain-shifted and  partially-annotated datasets.
Starting from a variational formulation of the joint problem, we show how the expected risk can be decomposed and optimised empirically.
We exploit an upper bound of the risk to deal with heterogeneous imaging modalities across datasets. To deal with potential domain shift, we integrated and tested three conventional techniques based on data augmentation, adversarial learning and pseudo-healthy generation.
For each individual task, our joint approach reaches comparable performance to task-specific and fully-supervised models.
The proposed framework is assessed on two different types of brain lesions: White matter lesions and gliomas.
In the latter case, lacking a joint ground-truth for quantitative assessment purposes, we propose and use a novel clinically-relevant qualitative assessment methodology.
%%%%
\end{abstract}

\begin{keyword}
% \MSC 41A05\sep 41A10\sep 65D05\sep 65D17
%% Keywords
\KWD Joint Learning \sep Domain Adaptation \sep Multi-Task Learning \sep Multi-Modal  
\end{keyword}
\end{frontmatter}

%\linenumbers


\section{Introduction}
Traditional approaches for tissue segmentation used in brain segmentation / parcellation software packages such as FSL \citep{Jenkinson:NeuroImage:2012}, SPM \citep{Ashburner:NeuroImage:2000} or NiftySeg \citep{Cardoso:TMI:2015} are based on subject-specific optimisation. FSL and SPM fit a Gaussian Mixture Model to the MR intensities using either a Markov Random Field (MRF) or tissue prior probability maps as regularisation. Alternatively, multi-atlas methods rely on label propagation and fusion from multiple fully-annotated images, i.e. atlases, to the target image \citep{Cardoso:TMI:2015,Iglesias:MedIA:2015}. These methods typically require extensive pre-processing, e.g. skull stripping, correction of bias field and registration. They are also often time-consuming and are inherently only adapted for brains devoid of large anatomical changes induced by pathology, such as white matter lesions and brain tumours.
Indeed, it has been shown that the presence of lesions can significantly distort any registration output \citep{Sdika:HBM:2009}. Similarly, lesions introduce a bias in the MRF. This leads to a performance degradation in the presence of lesions for brain volume measurement \citep{Battaglini:HBM:2012} and any subsequent analysis.

While quantitative analysis is expected to play a key role in improving the diagnosis and follow-up evaluations of patients with brain lesions, current tools mostly focus on quantification of the lesions themselves and effectively discard contextual tissue information.
Existing quantitative neuroimaging approaches allow the extraction of imaging biomarkers such as the largest diameter, volume, and count of the lesions. Such automatic segmentation of the lesions promises to speed up and improve the clinical decision-making process but more refined analysis would be feasible from tissue classification and region parcellation. In particular,  brain atrophy at a global level~\citep{Popescu1082,doi:10.1002/jmri.23671}, at a cerebral level~\citep{BERMEL2006158}, and, even more specifically, at the grey matter level~\citep{GEURTS20121082} have been correlated with the speed of disease progression and with physical disability~\citep{doi:10.1002/ana.21423}. Consequently, atrophied tissue volumes in the presence of lesions are clinically relevant imaging markers~\citep{doi:10.1111/jon.12527}.
We believe that, although very few work have addressed this problem yet, a joint model for lesion and tissue segmentation is expected to bring significant clinical benefits.
As representative exemplars of the technical challenges involved to build such joint models, we focus, in this work, on patients with white matter lesions or brain tumours.

Deep Neural Networks (DNNs) have become the state-of-the-art for most segmentation tasks \citep{DBLP:journals/corr/abs-1902-09063} and one would now expect these to be able to jointly perform brain tissue and pathology segmentation. However, annotated databases required to train DNNs are usually dedicated to a single task (either brain tissue segmentation or pathology delineation). In addition, the information required for brain tissue or pathology segmentation may come from different scans, leading to hetero-modal (i.e. more than one set of input imaging sequences) datasets. While \Tone-weighted images provide the best grey/white matter contrast for the delineation of anatomical tissue, \Ttwo-weighted sequences are usually more sensitive to pathological changes \citep{doi:10.1148/rg.262055063}. Choice of the used sequence or combination of sequences may also differ across pathologies. \Ttwo-weighted FLAIR images are often used for the assessment of white matter lesions  \citep{Maillard54} while a combination of \Tone contrast-enhanced (\Tonec), \Ttwo and FLAIR is often preferred for the characterisation of gliomas \citep{doi:10.1200/JCO.2009.26.3541}.
Similarly, the scans may have been acquired with different magnetic resonance parameters leading to differences in resolution and contrast among databases. Consequently, the data distribution may differ between the datasets, i.e. the datasets may be domain-shifted.
Given 1) the large amount of resources, time and expertise required to annotate medical images, 2) the varying imaging requirement to support each individual task, and 3) the availability of task-specific databases, it is unlikely that large databases for every joint problem, such as lesion and tissue segmentation, will become available.
%for research purposes.
There is thus a need to exploit existing task-specific databases to address the joint problems. Learning a joint model from task-specific hetero-modal and domain-shifted datasets is nonetheless challenging. As shown in Figure~\ref{fig:overview}, this problem lies at the intersection of Multi-Task Learning \citep{DBLP:journals/corr/ZhangY17aa}, Domain Adaptation \citep{Ben-David2010,DBLP:conf/icml/0002CZG19} and Weakly Supervised Learning \citep{Oquab_2015_CVPR,Bilen_2016_CVPR,XU2014591} with singularities making individual methods from these underpinning fields insufficient to address it completely, as explained in more depth in the related work section \ref{sec:related_work}.

Our approach is rooted in all these sub-domains of deep learning. The main contributions are summarised as follows:
\begin{enumerate}
    \item We propose a joint model that performs tissue and lesion segmentation as a unique joint task and thus exploits the interdependence between the lesion and tissue segmentation tasks. Starting from a variational formulation of the joint problem, we exploit the disjoint nature of the label sets to propose a practical decomposition of the joint loss, transforming the multi-class segmentation problem into a multi-task problem.
    \item We introduce feature channel averaging across modalities to adapt existing networks for our hetero-modal problem.
    \item We develop a new method to minimise the expected risk under the constraint of missing modalities.  Under the assumption that the network is not affected by a potential domain shift, we show that the expected risk can be further decomposed and minimised via a tractable upper bound. To our knowledge, no such optimisation method for missing modalities in deep learning has been published before.
    \item  Given that, in practice, the heterogeneous task-specific datasets may have been acquired with different protocols, i.e. they are domain-shifted, we integrate several existing DA techniques in our framework. These methods are based on data augmentation and adversarial training, or pseudo-healthy brain generation.
     \item We demonstrate the performance of our joint approach on two clinical use cases: White matter lesions and gliomas. Our method outperforms a fully-supervised model trained on a smaller fully-annotated dataset for white matter lesions. To assess the performance of the joint model for tissue and glioma segmentation, for which no ground-truth is available, we propose a new qualitative evaluation protocol based on the ASPECTS score \citep{BARBER20001670}. Higher accuracy is obtained compared to time-consuming pipelines that require to mask the lesions using manual annotations.
    \item Experiments show that generating  pseudo-healthy  annotated  scans outperforms the other DA techniques, even with very few pseudo-healthy annotated scans.
    % Experiments show that generated annotated healthy-looking brain is
    % Experiments show that generating annotated healthy-looking brain is more efficient to perform DA than data augmentation and adversarial learning, even with only very few samples.
\end{enumerate}

This work is a substantial extension of our conference paper~\citep{pmlr-v102-dorent19a}. Improvements include: 1) Additional mathematical proofs; 2) integration and validation of three different domain adaptation techniques to cope with domain-shifted datasets; 3) new experiments on joint brain tissue and glioma segmentation; and 4) a new quantitative evaluation protocol for assessing tissue segmentation in the absence of ground-truth.

\begin{figure}[tb!]
 % Caption and label go in the first argument and the figure contents
 % go in the second argument
  \includegraphics[width=\linewidth]{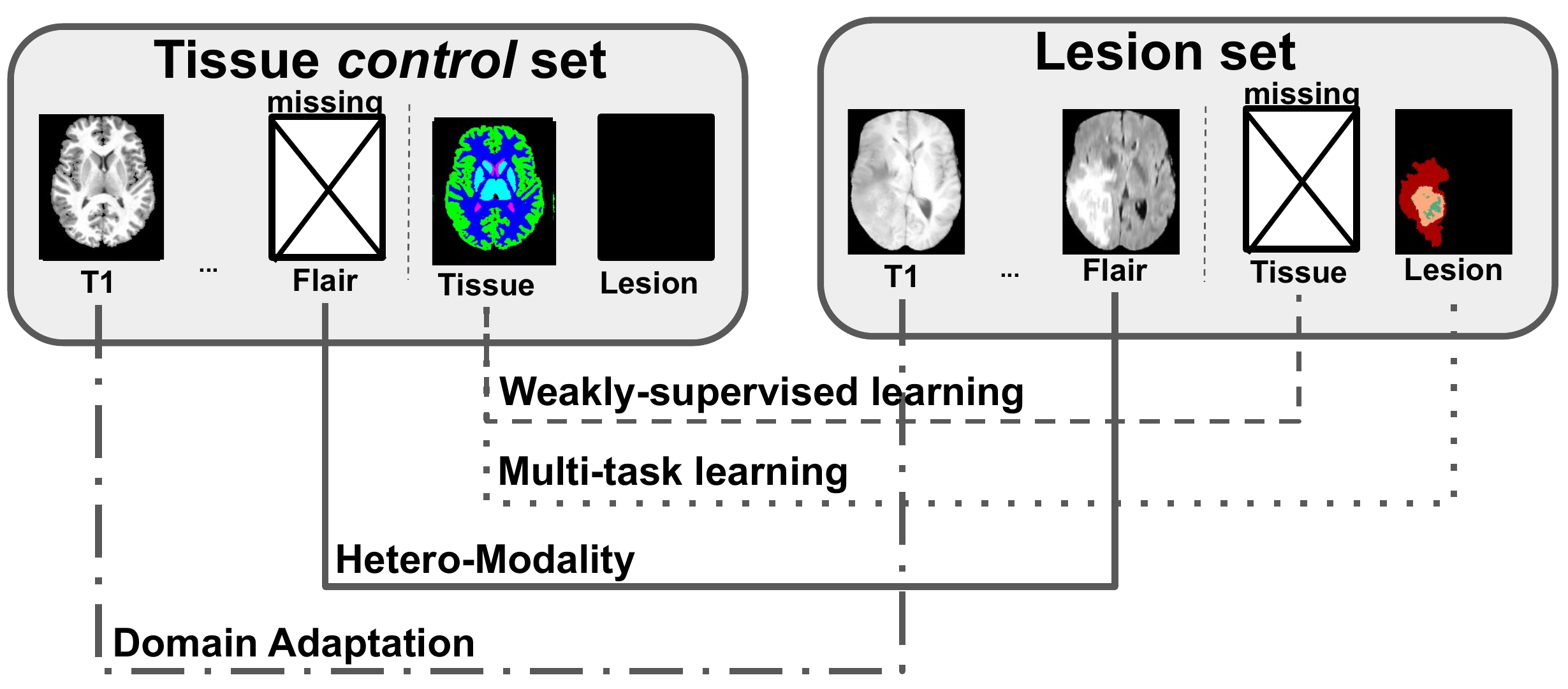}
  \caption{Tissue and lesion segmentation, a problem at the intersection of multiple branches of Machine Learning: Multi-Task Learning (tissue + lesion segmentation), Weakly-Supervised Learning (missing annotations), Hetero-Modality (missing modalities), Domain Adaptation (different acquisition protocols). }
  \label{fig:overview}
\end{figure} 

% First Secondly,  Thirdly,  Finally, we evaluate our tissue and lesion segmentation framework for two types of lesions: White matter lesions and glioma.  We demonstrate that our joint approach can achieve, for each individual task, similar performance compared to a task-specific baseline.
% %
% Albeit relying on different annotation protocols, results using a small fully-annotated joint dataset demonstrate efficient generalisability.
\begin{figure*}[tb!]
 % Caption and label go in the first argument and the figure contents
 % go in the second argument
  \centering
  \includegraphics[width=\linewidth]{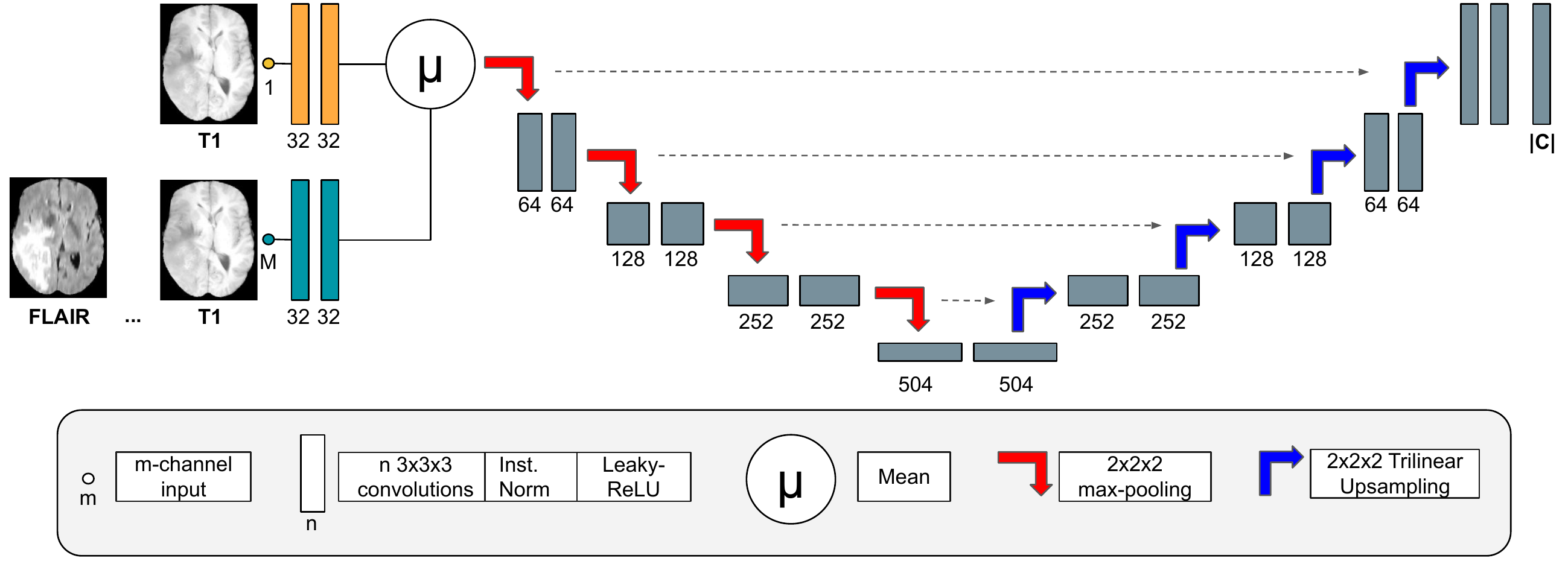}
\caption{The proposed fully-convolutional network architecture: A mix a 3D U-Net \citep{Unet} and HeMIS \citep{Havaei:MICCAI:2016}. }
  \label{fig:network}
\end{figure*}

\section{Related work}\label{sec:related_work}
Multi-Task Learning (MTL)  aims  to  perform  several tasks simultaneously, on a single dataset, by extracting some form of common knowledge or representation and introducing a task-specific back-end. When relying on DNN for MTL, the first layers of the network are typically shared, while the last layers are trained for the different tasks \citep{DBLP:journals/corr/Ruder17a}. MTL has been successfully applied to medical imaging for segmentation \citep{Bragman:MICCAI:2018, MTL:MICCAI:2016} combined with other tasks such as detection \citep{MTL:MIDL19:Saha} or classification \citep{MTL:MICCAI:2019,MTL:MIDL19:Le}. 
The global loss function is a weighted sum of task-specific loss functions. Recently, \citet{Kendall:NIPS:2017} proposed a Bayesian parameter-free method to estimate the MTL loss weights and \citet{Bragman:MICCAI:2018} extended it to spatially adaptive task weighting and applied it to medical imaging. Although the aforementioned approaches generate different outputs from the same features, no direct interaction between the task-specific outputs is modelled in these techniques. While a joint tissue and lesion segmentation can be pursued in practice, a strong underpinning assumption is that the two outputs are conditionally independent. Consequently, these approaches do not address the problem of aggregating these outputs to generate a joint segmentation.
Moreover, MTL approaches, such as \citep{MOESKOPS2018251,pmlr-v102-roulet19a}, do not provide any mechanism for dealing with hetero-modal datasets or changes in imaging characteristics across task-specific databases.

Domain Adaptation (DA) is a solution for dealing with domain-shifted datasets, i.e. datasets acquired with different settings. 
A classical strategy consists in learning a domain-invariant feature representation of the data. \citet{Csurka:DACVARevChap:2017} proposed an extensive review of these methods in deep learning. 
Some DA approaches have been developed to tackle a specific and identified shift.
For example, data augmentation has been used for shifts caused by different MR bias fields \citep{SUDRE201750} or the presence of motion artefacts \citep{RICHARD};
\citet{Havaei:MICCAI:2016} and \citet{Dorent:MICCAI:2019} proposed network architectures for dealing with missing modalities that encodes each modality into a shared modality-agnostic latent space.
Recent studies have proposed to learn a mapping between healthy and decease scans, using CycleGANs \citep{GAN:LesionRemoval1,GAN:LesionRemoval2} or Variational Autoencoders \citep{DetectionLesionUnsupervised}. Although these techniques have shown promising results, they are inherently limited to a specific type of shift. Combining causes of shift, for instance the presence/absence of lesions with different protocols of acquisition, remains an unsolved problem.
In contrast, general DA approaches do not make assumptions about the nature of the shift. These methods aim to directly minimise the discrepancy between the feature distributions across the domains. Distribution dissimilarity can be assessed using correlation distances \citep{Correlation} or maximum mean discrepancy \citep{MMD1,MMD2}. However, more recent techniques are mostly focused on adversarial methods, achieving promising results in medical imaging \citep{Kamnitsas:MICCAI:2017,Unsupervised-Cross-Modality,OrbesArteaga2019MultiDomainAI}. However, these methods are usually focus on solving a single task across domain.

Weakly-supervised Learning (WSL) deals with missing, inaccurate, or inexact annotations. Our problem is a particular case of learning with missing labels since each task-specific dataset provides a set of labels where the two sets are disjoint. \citet{Li:PAMI:2017} proposed a method to learn a new task from a model trained on another task. This method combines DA through transfer learning and MTL. In the end, two models are created: One for the first task and one for the second one.  \citet{Kim:WACV:2018} extended this approach by using a knowledge distillation loss in order to create a unique multi-task model. This aims to alternatively learn one task without forgetting the other one. The WSL problem was thus decomposed into an MTL problem with aforementioned limitations for our specific use case.

This work proposes a new framework to perform a joint segmentation while dealing with task-specific, domain-shifted and hetero-modal datasets. 

%%%
\section{Tissue and lesion segmentation learning from hetero-modal and task-specific datasets: Problem definition}
In order to develop a joint model, we first propose a mathematical variational formulation of the problem and introduce a network architecture to leverage existing hetero-modal and task-specific datasets for tissue and lesion segmentation.

\subsection{Formal problem statement}
Let $x=(x^1,..,x^M) \in \mathcal{X}=\mathbb{R}^{N\times M}$ be a vectorized multimodal image and $y \in \mathcal{Y}=\{0,..,C\}^{N}$ its associated tissue and lesion segmentation map. $N$, $M$ and $C$ are respectively the number of voxels, modalities and classes. Note that images modalities are assumed to be co-registered and resampled in the same coordinate space containing $N$ voxels. Our goal is to determine a predictive function, parametrised by the weights $\theta\in\Theta$, $h_{\theta}:\mathcal{X} \mapsto \mathcal{Y}$ that minimises the discrepancy between the ground truth label vector $y$ and the prediction $h_{\theta}(x)$. Let $\mathcal{L}$ be a loss function that estimates this discrepancy. 
Following the formalism used by \citet{Bottou:SIAMRev:2018}, given a probability distribution $\mathcal{D}$ over $(\mathcal{X},\mathcal{Y})$ and random variables under this distribution, we want to find $\theta ^{*}$ such that:

\begin{equation} \label{eq:0}
 \theta^{*}=\argmin_{\theta} \mathbb{E}_{(x,y) \sim \mathcal{D}}\left[\mathcal{L}\left(h_{\theta}(x), y\right)\right] 
\end{equation}

As is the norm in data-driven learning, we do not have access to the true joint probability $\mathcal{D}$.  In supervised learning, the common method is to estimate the expected risk using training samples. Given a set of $n \in \mathbb{N}$ independently drawn multimodal scans with their associated tissue and lesion segmentation map $\left\{ \left( x_k,y_{k} \right) \right\}_{k=1}^{n}$, we want to find $\theta ^{*}$ that minimises the empirical risk:
\begin{equation} \label{eq:fullyestimation}
 \theta^{*}=\argmin_{\theta} \sum_{k=1}^{n} \mathcal{L}\left(h_{\theta} \left(x_k \right), y_{k}\right)
\end{equation}
However, in our  multitask scenario, we cannot directly estimate the empirical risk since we do not have access to a fully annotated dataset for the joint task. Instead, we propose to leverage task-specific and hetero-modal datasets.

\subsection{Task-specific and hetero-modal datasets}
Let us assume that we have access to two datasets with either the tissue annotations $y^{T}$ or the lesion annotations $y^{L}$ (task-specificity). Let 
\begin{align*}
\mathcal{S}_{control}&=\left\{ \left( \left(x_k^1,..,x_{k}^{M_{T}}\right),y^{T}_k \right)\right\}_{k=1}^{n_{T}} \\
\mathcal{S}_{lesion}&=\left\{ \left( \left(x_k^1,...,x_{k}^{M_L}\right),y^{L}_k \right)\right\}_{k=1}^{n_{L}}
\end{align*}
denote these two training sets, where $M_{T}$, $M_{L}$, $n_{T}$ and $n_{L}$ are respectively the number of modalities in the \emph{control} and the \emph{lesion} sets and the size of these sets. Note that although we use the term \emph{control} for convenience, we may expect to observe pathology with "diffuse" anatomical impact, e.g. from dementia. In addition, for the clarity of presentation, we highlight that the considered \emph{lesions} in this work are either White Matter Hyperintensities (WMH) or gliomas.

Since such datasets are typically developed in the scope of either tissue or lesion segmentation (but not both), the set of observed modalities may vary from one dataset to another (hetero-modality). Importantly, in this work, we consider that only \Tone scans are provided in the \emph{control} dataset, while the \emph{lesion} set contains either 1) the \Tone and the FLAIR scans for WMH segmentation, or 2) the \Tone, contrast-enhanced \Tone~(\Tonec), \Ttwo and FLAIR scans for glioma segmentation.
The full set of modalities is consequently given by the modalities in the \emph{lesion} set, while the \emph{control} dataset will have missing modalities.
%Note also that, 
In our specific use cases, the \Tone modality is a shared modality across the different datasets. It will nonetheless be apparent that our method can be trivially adapted for other shared modalities.

\subsection{On the distribution \texorpdfstring{$\mathcal{D}$}{D} in the context of heterogeneous databases}
As we expect different distributions across heterogeneous databases, two probability distributions of $(X,Y)$ over $(\mathcal{X},\mathcal{Y})$ can be distinguished:
\begin{itemize}
    \item under $\mathcal{D}_{control}$, $(X,Y)$ corresponds to a multimodal scan and joint segmentation map of a patient without lesions ($Y$ effectively being a tissue segmentation map).
    \item under $\mathcal{D}_{lesion}$, $(X,Y)$  corresponds to a multimodal scan and joint lesion and tissue segmentation map of a patient with lesions.
\end{itemize}

Since traditional tissue segmentation methods are not adapted in the presence of lesions, the most important and challenging distribution $\mathcal{D}$ to address is the one for patients with lesions, $\mathcal{D}_{lesion}$.
In the remainder of this work, we thus assume that:
\begin{equation}\label{eq:H2}\tag{$\textbf{H}_1$}
\mathcal{D} \triangleq \mathcal{D}_{lesion}.
\end{equation}

% Let additionally assume that the acquisition protocols of the \Tone scans are partially similar across the \emph{control} and \emph{lesion} dataset, meaning that two \emph{control} distributions can be distinguished:
% \begin{itemize}
%     \item the \Tone scans under $\mathcal{D}_{control}^{SP}$ and $\mathcal{D}_{lesion}$ have been acquired with a similar protocol (SP).
%     \item the \Tone scans under $\mathcal{D}_{control}^{DA}$ and $\mathcal{D}_{lesion}$ have been acquired with inconsistent protocols meaning that Domain Adaptation (DA) is required.
% \end{itemize}
% Thus, two sets of samples $\mathcal{S}_{control}^{SP}$ and $\mathcal{S}_{control}^{DA}$ will be observed associated respectively to two distributions $\mathcal{D}^{SP}_{control}$ and $\mathcal{D}_{control}^{DA}$. Due to the domain gap, a model trained on one dataset cannot be directly applied to the other without domain adaptation.

\subsection{Hetero-modal network architecture}
% In order to learn from hetero-modal datasets, we need a network architecture that allows for missing modalities. We propose an architecture inspired by HeMIS \citep{Havaei:MICCAI:2016} shown in Figure~\ref{fig:network}. Features of each modality are first extracted separately and are then averaged, such as in \citep{Havaei:MICCAI:2016}. The average operation is only performed on the features extracted from the observed modalities and thus allows for missing modalities. In order to preserve the spatial resolution between the input and the output, we proposed a fully-convolutional architecture without any down-sampling operation and used dilated convolutions to capture information at multiple scales, such as in \citet{Li:IPMI:2017}. Finally, residual connections \citep{7780459} are added to avoid the problem of vanishing gradients. This hetero-modal network with weights $\theta$ is used to capture the predictive function $h_{\theta}$ that can accept either $M^T$ or $M^L$ modalities as input.
In order to learn from hetero-modal datasets, we need a network architecture that allows for missing modalities. Specifically, the input modalities are either a \Tone scan or a full set of modalities. To deal with missing modalities, arithmetic operations are employed, as originally proposed in HeMIS \citep{Havaei:MICCAI:2016}. The network architecture is based on a U-Net \citep{3DUnet}, as shown in Figure~\ref{fig:network}. Note that, while the proposed method requires a hetero-modal network, any specific architecture can be used. 
The proposed network is composed of two input branches, one for the \Tone scan and one for the full set of modalities. Although HeMIS originally proposed to encode each modality independently, i.e one branch per modality, we experimentally found higher performance with these two branches. In the presence of the full set of modalities, features extracted from the \Tone scan and all the modalities are averaged. Consequently, the network allows for missing modalities, i.e. is hetero-modal. This hetero-modal network with weights $\theta$ is used to capture the predictive function $h_{\theta}$ that can accept either \Tone or the full set of modalities as input.

\section{Optimising tissue and lesion segmentation as a joint task}
Given the mathematical formulation of the problem and the hetero-modal network architecture, we propose a method to empirically optimise the joint problem of tissue and lesion segmentation.

\subsection{Loss decomposition}
Let $\mathcal{C}_T$, $\mathcal{C}_L$  and $0$ be respectively the set of tissue classes and lesion classes and the value of the background class in the segmentation masks. Since $\mathcal{C}_T$ and $\mathcal{C}_L$ are disjoint, the segmentation map $y$ can be decomposed into two segmentation maps $y = y^{L}+y^{T}$ with $y^{T} \in \mathcal{C}_T\cup \{0\}$, $y^{L} \in \mathcal{C}_L\cup \{0\}$.

% \begin{figure}[t]
%  % Caption and label go in the first argument and the figure contents
%  % go in the second argument
%   \caption{Decomposition of the segmentation map $y$ into a tissue segmentation map $y^T$ and a lesion segmentation map $y^L$}
%   \includegraphics[width=\linewidth]{figs/decomposition_seg.pdf}
%   \label{fig:decomposition}
% \end{figure} 

Let us assume that the loss function $\mathcal{L}$ can also be decomposed into a tissue loss function $\mathcal{L}^{T}$ and a lesion loss function $\mathcal{L}^{L}$. This is common for multi-class segmentation loss functions in particular for those with \emph{one-versus-all} strategies (e.g. Dice loss, Jaccard loss, Generalized Cross-Entropy). Then, the joint and multi-class segmentation problem can be formulated as a multi-task problem:
\begin{equation}\label{eq:H1}\tag{$\textbf{H}_2$}
\mathcal{L}(h_{\theta}(x), y) = \mathcal{L}^{T}(h_{\theta}(x), y^{T}) + \mathcal{L}^{L}(h_{\theta}(x), y^{L})
\end{equation}
In combination with \eqref{eq:H2}, \eqref{eq:0} can be rewritten as:
\begin{equation} \label{eq:1}
\theta^{*} = \text{argmin}_{\theta}  \underbrace{\mathbb{E}_{ \mathcal{D}_{lesion}}[\mathcal{L}^{T}(h_{\theta}(x), y^{T})]}_{\mathcal{R}^T} + \underbrace{\mathbb{E}_{\mathcal{D}_{lesion}}[\mathcal{L}^{L}(h_{\theta}(x), y^{L})]}_{\mathcal{R}^L}
\end{equation}
While the second expected risk $\mathcal{R}^L$ can be estimated using the full set of modalities and the lesion annotations provided in the \emph{lesion} dataset, the first expected risk $\mathcal{R}^T$ appears to be intractable due to the missing tissue annotations in the \emph{lesion} dataset. In the next sections, we first propose an upper bound of the expected tissue risk $\mathcal{R}^T$ and then a means to estimate this upper bound using the \emph{control} dataset.

\subsection{upper bound of the expected tissue risk \texorpdfstring{$\mathcal{R}^T$}{R\^{}T}}
Although, thanks to its hetero-modal architecture, $h_{\theta}$ may handle inputs with varying number of modalities, the current decomposition \eqref{eq:1} assumes that all the modalities of $x$ are available for evaluating the loss. In our scenario, the \emph{control} set of scans with tissue annotations only contains the \Tone scans. Consequently, as we do not have all the modalities with tissue annotations, and as naively evaluating a loss with missing modalities would lead to a bias, estimating $\mathcal{R}^{T}$ is not straightforward. 

Let us assume that the tissue loss function $\mathcal{L}^{T}$ satisfies the triangle inequality:
\begin{equation}\label{eq:H3}\tag{$\textbf{H}_3$}
\forall (a,b,c) \in \mathcal{Y}^{3}: \ \mathcal{L}^{T}(a,c) \leq \mathcal{L}^{T}(a,b) + \mathcal{L}^{T}(b,c)
\end{equation}
Although not all losses satisfy \eqref{eq:H3}, it is known that the binary Jaccard is a distance \citep{Spaeth:ORS:1981,Kosub:PRL:2018} and thus satisfies the triangle inequality.
\begin{definition} 
(Binary Jaccard distance)\\
The binary Jaccard distance $J_{bin}$ is defined such that:
\begin{equation}\label{eq:binary_jaccard}
    %\begin{split}
        \forall a,b \in \{0,1\}^{N}, \ J_{bin}(a,b) = 1- \frac{\sum_{i=1}^{N} a_i b_i}{\sum_{i=1}^{N} a_i + b_i-a_ib_i} 
    %\end{split}
\end{equation}%
\end{definition}
\noindent However, network outputs are typically pseudo-probabilities, and the soft version of \eqref{eq:binary_jaccard} does not satisfy the triangle inequality. To satisfy \eqref{eq:H3}, we extend the binary Jaccard distance to a multi-class probabilistic formulation that coincides with the binary Jaccard for binary inputs but preserves the metric property for probabilistic inputs.
%
% \begin{definition} 
% (Probabilistic multi-class Jaccard distance)

% Let $\norm{.}_{1}$ denote the $L_1$ norm and  $\mathcal{P} \subset [0,1]^{N\times C} $ the set of probability vectors such that for any $\mathbf{p}=(p_c)_{c\in \mathcal{C}} \in \mathcal{P}$:
%     $$\forall i \in [0,..N], \ \sum_{i\in C} p_{i,c}=1 $$
%     The probabilistic multi-class Jaccard distance is defined for any $(\mathbf{u},\mathbf{v})\in\mathcal{P}^2$ as:
% \begin{equation*}\label{eq:jaccard_prob}
% \begin{split}
%     \mathcal{J}(\mathbf{u},\mathbf{v}) =\sum_{c \in C} \omega_{c} \frac{2\norm{u_c-v_c}_{1}}{\norm{u_c}_{1}+\norm{v_c}_{1}+\norm{u_c-v_c}_{1}} \ \ \text{with} \sum_{c \in C} \omega_{c} =1 
% \end{split}
% \end{equation*}
% \end{definition}
%
\begin{definition} 
(Probabilistic multi-class Jaccard distance)\\
Let $C$ be the number of classes in $\mathcal{C}$, $N$ be the number of voxels and $\mathcal{P} \subset [0,1]^{C\times N} $ denote the set of probability vector maps such that for any $p=(p_{c,i})_{c\in \mathcal{C},~ i\in[0;N]} \in \mathcal{P}$:
    $$\forall i \in [0;N], \ \sum_{c\in C} p_{c,i}=1 $$
    The probabilistic multi-class Jaccard distance is defined for any $(u,v)\in\mathcal{P}^2$ as:
\begin{equation}\label{eq:jaccard_prob}
\begin{split}
    %\forall (u,v)\in \mathcal{P}^2, \
    \mathcal{J}(u,v) =\sum_{c \in C} \omega_{c} \frac{2\sum_{i=1}^{N}|u_{c,i}-v_{c,i}|}{\sum_{i=1}^{N}|u_{c,i}|+|v_{c,i}|+|u_{c,i}-v_{c,i}|} 
\end{split}
\end{equation}
where $\omega_c$ are the class weights such that $\text{with} \sum_{c \in C} \omega_{c} =1$
\end{definition}
\noindent As shown in \ref{appendix:relationship}, the binary and probabilistic Jaccard distance coincide on the set of binary vectors $\{0,1\}^{N}$.
Furthermore, \eqref{eq:jaccard_prob} satisfies \eqref{eq:H3}.
\begin{lemma}\label{lemma:probabilisticjaccard}
    The probabilistic multi-class Jaccard distance is a distance and thus satisfies the triangle inequality.
\end{lemma}
\begin{proof}
The proof, detailed in \ref{appendix:proofJaccarddistance}, follows from the Steinhaus transform \citep{Spaeth:ORS:1981} applied to the metric space $([0,1]^{N}, d_{1})$ where $d_{1}$ is the distance induced by the $L_1$ norm. 
\end{proof}

Under \eqref{eq:H3}, $\mathcal{L}^{T}$ satisfies the following inequality:
\begin{equation}
\mathcal{L}^{T}(h_{\theta}(x), y^{T}) \leq \mathcal{L}^{T}(h_{\theta}(x), h_{\theta}(x^{T_1})) + \mathcal{L}^{T}(h_{\theta}(x^{T_1}), y^{T})
\end{equation}
where $x^{T_1}$ denotes the \Tone scan associated to $x$.
Consequently, we find an upper bound of the expected tissue risk:
\begin{equation}\label{eq:upbnd_init}
    \begin{split}
    \mathcal{R}^{T}(\theta) \leq  \underbrace{\mathbb{E}_{ \mathcal{D}_{lesion}}[\mathcal{L}^{T}(h_{\theta}(x), h_{\theta}(x^{T_1}))] }_{\mathcal{R}^{T}_{T_1 \to Full}}
    + \underbrace{\mathbb{E}_{ \mathcal{D}_{lesion}}[\mathcal{L}^{T}(h_{\theta}(x^{T_1}), y^T)]}_{\mathcal{R}^{T}_{T_1}}
    \end{split}
\end{equation}
Minimising $\mathcal{R}^{T}_{T_1}$ enforces the network to generate accurate tissue segmentation using only \Tone as input. Minimising $\mathcal{R}^{T}_{T_1 \to Full}$ encourages consistency between the outputs when given only \Tone or the full set of modalities as input. This latter term allows for transferring the knowledge learnt on the \Tone scan to the full set of modalities. 
%{\color{blue}Note that this term is not present in previously proposed hetero-modal networks, such as \cite{Havaei:MICCAI:2016} and \cite{Dorent:MICCAI:2019}}.

An empirical estimation of $\mathcal{R}^{T}_{T_1 \to Full}$ can be obtained by comparing the network outputs using either \Tone or the full set of modalities as input. In contrast, $\mathcal{R}^{T}_{T_1}$ requires comparison of inference done, under $\mathcal{D}_{lesion}$, from \Tone inputs with ground truth tissue maps $y^T$. While this provides a step towards a practical evaluation of $\mathcal{R}^{T}$, we still face the challenge of not having tissue annotations $y^T$ under $\mathcal{D}_{lesion}$.

\subsection{Estimating \texorpdfstring{$\mathcal{R}^T_{T_1}$}{R\^{}T\_T1} using the \emph{control} dataset}\label{section:withoutda}
To estimate $\mathcal{R}^{T}_{T_1}$, we assume that the neural network $h_{\theta}$ is invariant to a potential domain shift between the \Tone scans of the \emph{control} and \emph{lesion} datasets on the non-lesion regions. Specifically, we assume that the restriction of the feature distributions (rather than the image intensity distributions) over $\mathcal{D}_{lesion}$ and $\mathcal{D}_{control}$ to the non-lesion parts of the brain (i.e. the voxels $i$ such that $y_i \in \mathcal{C}_{T}$) are comparable, i.e.:
\begin{equation}\label{eq:H4}\tag{$\textbf{H}_4$}
%\forall i \in[1..N], \
P_{ \mathcal{D}_{lesion}}(h_{\theta}(x^{T_1})_i, y_i|y_i \in \mathcal{C}_{T})=P_{ \mathcal{D}_{control}}(h_{\theta}(x^{T_1})_i, y_i|y_i \in \mathcal{C}_{T})
\end{equation}

This means that the neural network $h_{\theta}$ generates similar outputs on the non-lesion parts of the brain between the two datasets, leading to:
\begin{equation}\label{eq:expec_eq}
\mathcal{R}^{T}_{T_1} = 
\mathbb{E}_{\mathcal{D}_{lesion}}[\mathcal{L}^{T}(h_{\theta}(x^{T_1}), y^T)] = \mathbb{E}_{\mathcal{D}_{control}}[\mathcal{L}^{T}(h_{\theta}(x^{T_1}), y^T)]
\end{equation}
Consequently, under \eqref{eq:H4}, $\mathcal{R}^{T}_{T_1}$ can be estimated using the \Tone scans and their tissue annotations in the \emph{control} dataset. Section~\ref{section:da} presents means of ensuring that assumption \eqref{eq:H4} is satisfied even in the presence of domain shift in the image intensity distributions.

\begin{figure*}[t!]
 % Caption and label go in the first argument and the figure contents
 % go in the second argument
  \centering
  \includegraphics[width=0.9\linewidth]{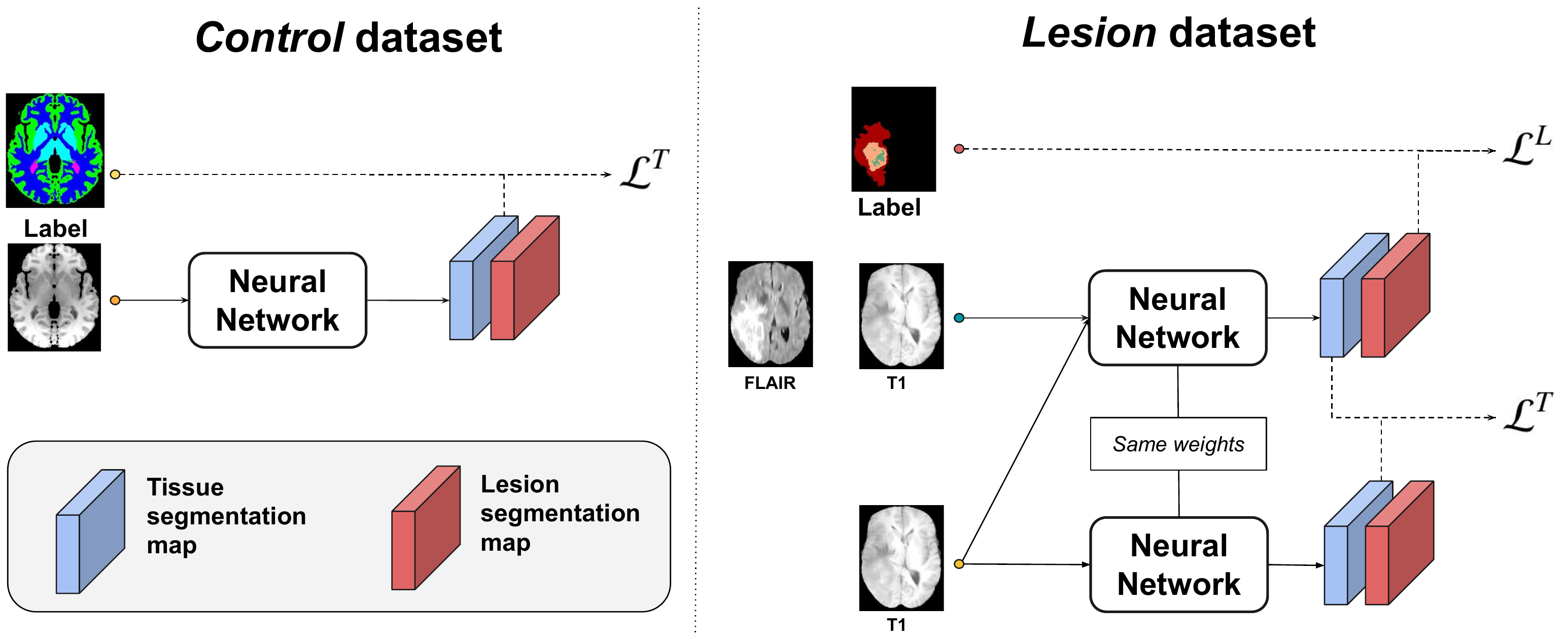}
   \caption{The training procedure using samples from the \emph{control} and \emph{lesion} datasets. The different elements of the decomposed loss upper bound are computed and minimised at each training iteration. The same network is used for all the different hetero-modal inputs. Note that domain adaptation is not represented.}
  \label{fig:training_procedure}
\end{figure*}

\subsection{Summary of the expected risk decomposition}
Bringing all the terms together, given  \eqref{eq:1}, \eqref{eq:upbnd_init} and \eqref{eq:expec_eq}, we seek the 
%optimal
parameters $\theta^{*}$ that optimise the tractable upper bound $\mathcal{R}_{seg}$ of the (intractable) expected risk:
\begin{equation} \label{eq:scenario1_final}
\begin{split}
    \theta^{*} = \text{argmin}_{\theta} & \left\{\mathcal{R}_{seg}= \mathbb{E}_{ \mathcal{D}_{control}}[\mathcal{L}^{T}(h_{\theta}(x^{T_1}), y^T)] + \right. \\
    & \left. \mathbb{E}_{ \mathcal{D}_{lesion}}[\mathcal{L}^{L}(h_{\theta}(x), y^{L})+\mathcal{L}^{T}(h_{\theta}(x), h_{\theta}(x^{T_1}))] \right\}\\
    % & \left. + \mathbb{E}_{ \mathcal{D}_{control}}[\mathcal{L}^{T}(h_{\theta}(x^{T_1}), y^T)]  \right\}
\end{split}
\end{equation}

\section{Matching feature distributions across datasets }\label{section:da} 
In this section, we explore different approaches that ensure the feature distributions extracted from the \emph{control} and \emph{lesion} \Tone scans are comparable, i.e. we want to satisfy \eqref{eq:H4} even in the presence of domain shift.
%Two scenarios are considered: The absence or the presence of domain shift between the \emph{control} and \emph{lesion} datasets.

\subsection{Similar acquisition protocols for the \emph{control} and \emph{lesion} datasets}\label{sec:scenario_noDA}
Let's first assume that the acquisition protocols are similar for the \emph{control} and \emph{lesion} datasets, i.e. they are not domain-shifted.  Specifically, we assume that the \Tone images have been acquired with similar sequences, spacial resolution and field strength. 
In this case, similar to \cite{DetectionLesionUnsupervised}, the restriction of the distributions $\mathcal{D}_{lesion}$ and $\mathcal{D}_{control}$ to the non-lesion parts of the brain can be assumed to be the same on the \Tone scans, i.e.:
\begin{equation}
%\forall i \in[1..N], \
P_{ \mathcal{D}_{lesion}}(x^{T_1}_{i},y_i|y_i \in \mathcal{C}_{T})=P_{ \mathcal{D}_{control}}(x^{T_1}_{i},y_i|y_i \in \mathcal{C}_{T})
\end{equation}

In the absence of domain shift, we can reasonably assume that the network produces similar outputs on the non-lesion parts of the brain for the two distributions, i.e. that \eqref{eq:H4} is satisfied. No specific additional action thus needs to be implemented.

\subsection{Generating pseudo-healthy scans to learn tissue segmentation from domain-shifted \Tone scans}\label{sec:pseudo_healthy}
Let's now consider the presence of a domain shift between the \Tone \emph{control} and \emph{lesion} scans due to different acquisition protocols. In this section, we propose to synthesise pseudo-healthy scans from domain-shifted \Tone lesion scans in order to extend the \emph{control} dataset with \emph{control} scans associated to the protocol of acquisition of the \textit{lesion} dataset. Since the \emph{control} and \emph{lesion} datasets are domain-shifted, existing lesion removal approaches, based either on CycleGANs \citep{GAN:LesionRemoval1,GAN:LesionRemoval2} or Variational Autoencoders \citep{DetectionLesionUnsupervised}, are not adapted as they require training data with no domain shift beyond the presence of absence of pathology.

To tackle the presence of an acquisition-related domain shift, we propose to generate pseudo-healthy scans and their annotations using traditional image computing techniques that are inherently robust to different acquisition protocols.
%a new method that is robust to domain shift and generates unpaired pseudo-healthy scans. 
For example, for white matter lesions, lesion filling methods allow for transforming scans with lesions into pseudo-healthy scans \citep{VALVERDE201486,PRADOS2016376}.
For large and unilateral pathology, we propose to synthesise pseudo-healthy \Tone scans by symmetrising the "healthy" hemisphere of the patients with lesions located in one hemisphere only.
The inter-hemispheric symmetry plane is estimated via the technique described in \citet{993131}. Finally, the "healthy" hemisphere of those patients is mirrored in order to create a symmetric pseudo-healthy brain.
Having generated pseudo-healthy images,
%Then,
traditional methods, designed for \emph{control} scans, such as the GIF framework \citep{Cardoso:TMI:2015}, can then be employed to generate the corresponding bronze standard tissue annotations.
%from these pseudo-healthy images.

With this set of scans  $S^{T_1}_{pseudo}$, we have access to a \emph{pseudo-control} dataset acquired with a similar protocol as in the \emph{lesion} dataset and similar on the non-lesion part of the brain, and thus are in the scenario described in \ref{sec:scenario_noDA}. Let denote  $\mathcal{D}^{T_1}_{pseudo}$ the distribution of those scans. The expected tissue risk $\mathcal{R}^{T}_{T_1}$ is then equal to the expect tissue risk under $\mathcal{D}^{T_1}_{pseudo}$:
\begin{equation}\label{eq:expect_synth}
\mathcal{R}^{T}_{T_1} = 
\mathbb{E}_{\mathcal{D}_{lesion}}[\mathcal{L}^{T}(h_{\theta}(x^{T_1}), y^T)] = \mathbb{E}_{\mathcal{D}^{T_1}_{pseudo}}[\mathcal{L}^{T}(h_{\theta}(x^{T_1}), y^T)]
\end{equation}
To take advantage of the manual annotations in the \emph{control} dataset, we resort to averaging the two formulations (\ref{eq:expec_eq}, \ref{eq:expect_synth}):
\begin{equation}\label{eq:expect_synth_final}
\mathcal{R}^{T}_{T_1} \approx \frac{\mathbb{E}_{\mathcal{D}_{control}}[\mathcal{L}^{T}(h_{\theta}(x^{T_1}), y^T)] + \mathbb{E}_{\mathcal{D}^{T_1}_{pseudo}}[\mathcal{L}^{T}(h_{\theta}(x^{T_1}), y^T)]}{2} 
\end{equation}
Consequently, the expected tissue risk $\mathcal{R}^{T}_{T_1}$ can be estimated using the \emph{control} and \emph{pseudo-control} \Tone scans.

\subsection{Alternative unsupervised DA techniques}
In order to satisfy \eqref{eq:H4}, the feature representation of the non-lesion parts of the brain has to be invariant to the changes induced by the different protocols. A direct way to align the feature distributions restricted to the non-lesion parts would be to match the representations of pairs of scans acquired with the different settings. However, in our scenario, we do not have access to such pairs of domain-shifted scans. 

In contrast, unsupervised DA allows to perform domain adaptation using unpaired and non-annotated domain-shifted scans. Unsupervised DA techniques commonly introduce an additional term ($\mathcal{R}_{DA}$) that encourages the network to be invariant to the domain shift. Then, the total expected risk reads:
\begin{equation}
    \mathcal{R}_{total}= \mathcal{R}_{seg} + \lambda \mathcal{R}_{DA}
\end{equation}
where $\lambda$ is a hyper-parameter that allows for balancing the segmentation risk $\mathcal{R}_{seg}$ (\eqref{eq:scenario1_final}) with the DA regularisation $\mathcal{R}_{DA}$.

The definition of the DA term depends on the DA technique. In this work, two common unsupervised DA methods are considered, based either on data augmentation or adversarial learning.

\subsubsection{Unsupervised DA via physically-inspired data augmentation}
Since \Tone scans play a key role for structure analysis, we expect high-resolution \Tone scans for datasets developed in the scope of tissue segmentation, such as the \emph{control} dataset. Conversely, \Tone scans are often less critical for lesion segmentation and \Tone scans may have been acquired with a lower resolution.

Let's assume that, less effort has been done to acquire high-resolution \Tone \emph{lesion} scans, explaining the differences in acquisition protocols. Specifically, we assume that the domain shift is caused by the presence of \Tone \emph{lesion} scans with artefacts (e.g. related to the MR bias field or the presence of motion artefacts) and a lower acquisition resolution. We additionally assume that differences of scanner characteristics (manufacturer, field strength) are excluded. 

%In this scenario, \eqref{eq:H4} is satisfied if the network $h_{\theta}$ is invariant to the aforementioned domain changes.

Then, physically-informed augmentation such as random bias field \citep{SUDRE201750} and motion artefacts \citep{RICHARD} and spacial smoothing can be employed to generate scans that are similar to the \Tone \emph{lesion} scans. Let denote $T_{\psi}$ the composition of these transformations parametrised by the parameters $\psi\sim\mathcal{D}_{\psi}$. For any \Tone \emph{control} scan $x^{T_1}_{c}$, we can thus generate an augmented version $T_{\phi}(x^{T_1}_{c})$, i.e. getting access to pairs of domain-shifted \Tone scans. This allows to minimise the discrepancy between the feature representations learnt by the neural network across the two domains by enforcing consistency across outputs from paired domain-shifted inputs, i.e.:
\begin{equation}\label{eq:augmentation}
\mathcal{R}_{DA}= \mathbb{E}_{\mathcal{D}_{\psi}}\mathbb{E}_{D_{control}}\left[ \mathcal{L}^{T}\left(h_{\theta}(x^{T_1}_{c}), h_{\theta}\circ T_{\phi}(x^{T_1}_{c})\right)\right]
\end{equation}
An empirical estimation of the DA regularisation term $\mathcal{R}_{DA}$ is obtained by comparing the network outputs using the \Tone \emph{control} scans and their augmented versions as input.

Consequently, if the domain shift is due to different spatial resolutions and the presence of the aforementioned artefacts, the network can be trained to be invariant to the domain shift, i.e. to satisfy \eqref{eq:H4}.

\subsubsection{Unsupervised DA via adversarial learning}
Let's now assume that the domain shift cannot easily be simulated. In this case, we can use adversarial learning. Adversarial approaches for domain adaptation can be seen as a two-player game: A discriminator $D_{\phi}$, parametrised by the weights $\phi\in\Phi$, is trained to distinguish the source domain features from the target domain features, while the segmentation network $h_{\theta}$ is simultaneously trained to confuse the domain discriminator. 

The discriminator aims to predict the probability that extracted features are part of the \emph{lesion} feature distributions.
The discriminator accuracy can thus be seen as a measurement of the discrepancy between the \emph{lesion} and \emph{control} feature distributions and used as a DA regularisation term:
%The discriminator accuracy is used to assess the discrepancy between the feature distributions, i.e.:
\begin{equation}\label{eq:loss_DA_unsupervised}
    \mathcal{R}_{DA}(\phi,\theta) = \mathbb{E}_{\mathcal{D}_{lesion}}[1-D_{\phi}(h_{\theta}(x))]+\mathbb{E}_{\mathcal{D}_{control}}[D_{\phi}(h_{\theta}(x))]
\end{equation}
This DA term can be estimated by using features extracted from \Tone \emph{control} and \emph{lesion} scans as input of the discriminator.

The following proposition shows that the discriminator accuracy is a principled measurement of the feature distribution discrepancy:
\begin{prop}[]
\label{theorem:domain-adaptation} Let assume that $\mathcal{L}$ satisfies the triangle inequality and is bounded. Let us also assume that the family of domain discriminators $\mathcal{H}_{\Phi}=\{D_{\phi}\}_{\phi \in \Phi}$ is rich enough. Then there is a constant $K$ such that:

\begin{equation}\label{eq:DA_unsupervised}
 \mathbb{E}_{\mathcal{D}_{lesion}}\left[\mathcal{L}\left(h_{\theta}(x), y\right)\right]  \leq \mathcal{R}_{seg}+ K\sup_{\phi}\mathcal{R}_{DA}(\phi,\theta)  + \epsilon(\Theta) 
\end{equation}
where $\epsilon(\Theta)$ is independent of the network parameters $\theta$ and corresponds to the accuracy of the best (and unknown) segmenter in the family of functions parametrised in $\Theta$.
\begin{proof}
    The proof uses \eqref{eq:1}, \eqref{eq:upbnd_init}, is based on \cite{Ben-David2010} and \cite{NIPS2018_7436} and detailed in \ref{appendix:prop}.
\end{proof}
\end{prop}
\eqref{eq:DA_unsupervised} shows that the intractable expected loss is bounded by a weighted sum of the tractable segmentation risk \eqref{eq:scenario1_final} and the accuracy of the best discriminator, up to a constant w.r.t the network parameters. 

Moreover, the alternative optimisation strategy can be seen as a way to estimate the best discriminator while minimising the upper bound defined in \eqref{eq:DA_unsupervised}. 

Note that \eqref{eq:DA_unsupervised} stands for features extracted at any level of $h_{\theta}$. In this work and similarly to \cite{OrbesArteaga2019MultiDomainAI}, the contracting path features from the U-Net are used as input of the discriminator.

\section{Implementation of the joint model optimisation}
Given the formulation of the joint model and our proposed computationally tractable decomposition, we present in this section the implementation of our framework. 

\subsection{Stochastic optimisation of the joint model}
We use a stochastic gradient descent approach to minimise the expected risk decomposition \eqref{eq:scenario1_final} and to enforce the network to be invariant to a potential domain shift between the datasets. The total loss function reads:
\begin{equation}
    \mathcal{L}_{total}= \mathcal{L}_{seg} + \lambda \mathcal{L}_{DA}
\end{equation}
where $\lambda$ is a hyper-parameter that allows for balancing the segmentation loss $\mathcal{L}_{seg}$ (associated to $\mathcal{R}_{seg}$) with the domain adaptation loss $\mathcal{L}_{DA}$ (associated to $\mathcal{R}_{DA}$) .
Figure~\ref{fig:training_procedure} shows the training procedure without DA. The weights of the segmentation loss are given by the decomposition of the problem. The domain adaptation parameter $\lambda$ is a hyper-parameter that is experimentally chosen.

At each training iteration, we draw pairs of samples $(x_{l},y^L_{l})$ and $(x^{T_1}_{c},y^{T}_{c})$ from $\mathcal{S}_{lesion}$ and $\mathcal{S}_{control}$ and compute in each mini-batch the following loss functions and associated gradient. Note that there is no natural pairing between $(x_{l},y^{L}_{l})$ and $(x^{T_1}_{c},y^{T}_{c})$. Our paired sampling procedure thus exploits random pairing.

As presented in Section~\ref{section:da}, different scenarios are considered.

\noindent \textbf{Similar acquisition protocols} If the datasets are not domain-shifted, no DA is required ($\lambda=0)$, and the segmentation loss is:
\begin{equation}\label{eq:seg_loss}
    \begin{split}
        \mathcal{L}_{seg} &=\mathcal{L}^L\left(h_{\theta}(x_{l}),y^{L}_{l}\right) + \mathcal{L}^{T}\left(h_{\theta}(x_{l}),h_{\theta}(x^{T_1}_{l})\right) \\
        & + \mathcal{L}^{T}\left(h_{\theta}(x^{T_1}_{c}), y^{T}_{c}\right) 
    \end{split}
\end{equation}
We experimentally found that the inter-modality tissue loss $\mathcal{R}^{T}_{T_1 \to Full}$ has to be skipped for few epochs (50 in our experiments).

% With domain-shifted databases, different techniques to adapt $h_{\theta}$  to  become invariant to acquisition protocols have been proposed:

\noindent \textbf{Pseudo-healthy generation} Given a pseudo-healthy annotated set of scans $S^{T1}_{pseudo}=\{x_{pseudo}^{T_1}, y_{pseudo}^{T}\}$, no DA is employed ($\lambda=0)$, and the segmentation loss, defined by \eqref{eq:expect_synth_final}, is:
\begin{equation}
    \begin{split}
        \mathcal{L}_{seg} &=\mathcal{L}^L\left(h_{\theta}(x_{l}),y^{L}_{l}\right) + \mathcal{L}^{T}\left(h_{\theta}(x_{l}),h_{\theta}(x^{T_1}_{l})\right) \\
        & + \frac{1}{2}\left[\mathcal{L}^{T}\left(h_{\theta}(x_{pseudo}^{T_1}),  y_{pseudo}^{T}\right) + \mathcal{L}^{T}\left(h_{\theta}(x^{T_1}_{c}), y^{T}_{c}\right)\right]
    \end{split}
\end{equation}

\noindent \textbf{DA via augmentation} If we assume that the differences of protocols can be simulated (random bias field, motion artefacts and spatial smoothing), the domain invariance \eqref{eq:augmentation} is learnt by minimising the inter-domain feature discrepancy defined as:
$$ \mathcal{L}_{DA} = \mathcal{L}^{T}\left(h_{\theta}(x^{T_1}_{c}), h_{\theta}(T_{\phi}(x^{T_1}_{c}))\right) $$
where $T_{\phi}$ corresponds to a composition of theses transformations. The segmentation loss $\mathcal{L}_{seg}$ is the same as in \eqref{eq:seg_loss}.

\noindent \textbf{DA via adversarial learning} If adversarial learning is employed, a discriminator $D$ is trained to discriminate scans from the two domains by maximising the domain classification accuracy. For computational stability, the
$L_1$ distance defined in \eqref{eq:loss_DA_unsupervised} has been replaced by the cross-entropy. Conversely, the segmenter $h_{\theta}$ is train to minimise this domain classification accuracy, i.e.:
$$ \mathcal{L}_{DA} = \log(D_{\psi}(h_{\theta}(x_{c}^{T_1})) + \log(1-D_{\psi}(h_{\theta}(x^{T_1}_{l})) $$
As in \cite{Kamnitsas:MICCAI:2017,OrbesArteaga2019MultiDomainAI}, the DA loss is skipped for few epochs (20 in our experiments) in order to initialise the discriminator. The segmentation loss $\mathcal{L}_{seg}$ is the same as in \eqref{eq:seg_loss}. 

\subsection{Implementation details}
We implemented our network in PyTorch, using TorchIO \citep{fern2020torchio}. 
%A Tensorflow version will also be released using NiftyNet \citep{Gibson:CMPB:2018}. 
Codes are available at \url{http://github.com/ReubenDo/jSTABL}.

Convolutional layers are initialised such as proposed in \citet{He:ICCV:2015}. The scaling and shifting parameters in the batch normalisation layers were initialised to 1 and 0 respectively. As suggested by \citet{Ulyanov:arXiv:2016}, we used instance normalisation. We used the same discriminator as in \cite{OrbesArteaga2019MultiDomainAI}.

 We performed a 3-fold cross validation. 
For each fold, we randomly split the data into $70\%$ for training, $10\%$ for validation and $20\%$ for testing. We used a batch of 2 \emph{lesion} scans, and 2 \emph{control} scans. Note that, for the DA approach based on data augmentation, the batch of 2 \emph{control} scans consists in a pair of non-augmented/augmented \emph{control} scans. As a data augmentation, a rotation with a random angle in $[-10^{\circ},10^{\circ}]$ and a random Gaussian noise are employed.
The network was trained using Adam optimiser \citep{Kingma:ICLR:2015}  the learning rates $l_R$, $\beta_1$, $\beta_2$ were initially respectively set up to $5.10^{-4}$, $0.9$ and $0.999$. $l_R$ was progressively reduced by a factor of 2 every 10000 iterations. 
We employed the training strategy used for the nnU-Net \citep{nnU-Net}: The learning rate is reduced by a factor 2 after 15 epochs without reduction of the exponential moving average of the loss on the validation split.

We used the probabilistic version of the multi-class Jaccard distance \eqref{eq:jaccard_prob} as the segmentation loss function.
In order to give the same weight to the lesion segmentation and the tissue segmentation, we choose $\bm{\omega}$ such that
\[
\sum_{c\in \mathcal{C}_{tissue}}\omega_{c} = \sum_{c'\in \mathcal{C}_{lesion}}\omega_{c'}=\frac{1}{2}.
\]

%%%
\begin{figure*}[tb!]
 % Caption and label go in the first argument and the figure contents
 % go in the second argument
  \centering
  \includegraphics[width=\linewidth]{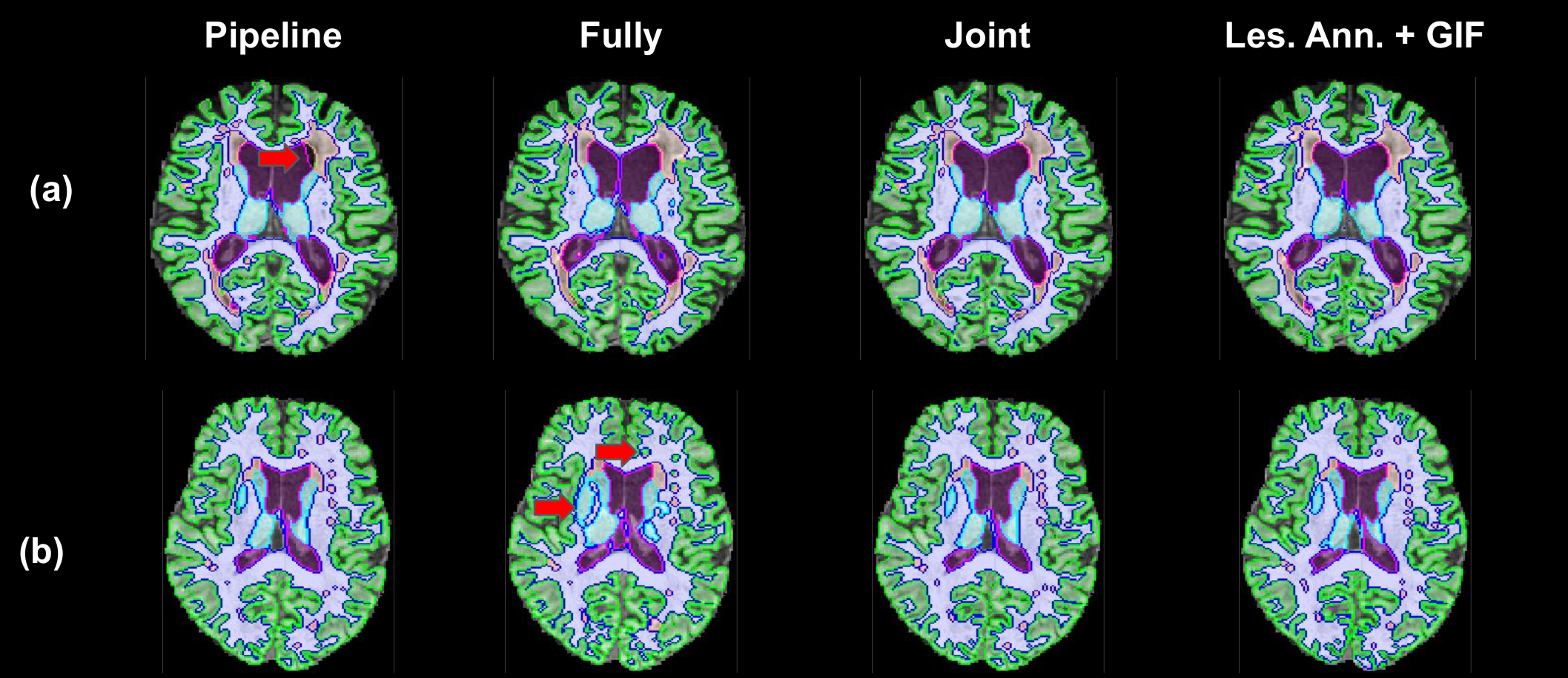}
    \caption{Examples of output for the different models: \emph{Pipeline}; \emph{Fully-sup}; \emph{jSTABL}; and the combination of the manual annotation and GIF. (a) \emph{Tissue} model used in \emph{Pipeline} can be perturbed by the presence of lesions (arrow). (b) Example for which the fully-supervised model largely fails to segment the tissue and the lesions.}
  \label{fig:resultsWMH}
\end{figure*} 
\begin{table*}[tb]
	\centering
	\caption{Summary of data characteristics for white matter lesion segmentation
	}\label{tab:DataCharacteristics}
	\resizebox{\textwidth}{!}{
	\begin{tabular}{c *{7}{c}}
		\toprule
		\multirow{2}{*}{} & \multicolumn{2}{c}{\bf Control data} & \multicolumn{3}{c}{\bf Lesion data } & \multicolumn{1}{c}{\bf Fully anno. data}\\
		               \cmidrule(lr){2-3} \cmidrule(lr){4-6} \cmidrule(lr){7-7}
		               & OASIS1 & ADNI2  & WMH-Utrech & WMH-Singapore & WMH-Amsterdam & MRBrainS18   \\
		               & \citep{doi:10.1162/jocn.2007.19.9.1498} & \citep{doi:10.1002/jmri.21049} & \multicolumn{3}{c}{\citep{8669968}} & \citep{MRBrainS18}\\
		\midrule
        \multirow{2}{*}{Sequences} &  3D MP-RAGE  \Tone  & 3D MP-RAGE \Tone & 3D MP-RAGE \Tone & 3D MP-RAGE \Tone & 3D MP-RAGE \Tone & MP-RAGE 3D \Tone \\ 
        & $\times$ & $\times$ & 2D FLAIR & 2D FLAIR & 3D FLAIR & 2D FLAIR  \\
        %\cmidrule(lr){2-7}
        MRI scanner & Siemens Vision & Various & Philips Achieva  & Siemens TrioTim & GE Signa HDxt & Philips Achieva  \\
        Field Strengh & 1.5T & 3T & 3T & 3T  & 3T & 3T  \\
        Voxel size ($\text{mm}^3$) &  $1.00\times1.00\times1.00$ & $1.20\times1.05\times1.05$ & $0.96\times0.95\times3.00$ & $1.00\times1.00\times3.00$ & $1.20\times0.98\times3.00$ & $0.96\times 0.96 \times 3.00$    \\
        %\cmidrule(lr){2-7}
        Annotations & 143 structures & 143 structures & WMH  & WMH & WMH & 6 Tissues+WMH+CSF \\
        $\#$ scans available & 35 & 25 & 20  & 20  & 20  & 30 (7 available)   \\
        \midrule
        Training data in: & \emph{Tissue} + \emph{jSTABL} & \emph{Tissue} + \emph{jSTABL} & \emph{Lesion} + \emph{jSTABL} & \emph{Lesion} + \emph{jSTABL} & \emph{Lesion} + \emph{jSTABL} & \emph{Fully} \\
		\bottomrule
		
	\end{tabular}
	}
	%\bigskip
\end{table*}

\section{Experiments and results}
\subsection{Joint white matter lesion and tissue segmentation}
\subsubsection{Task and datasets}\label{experiments:WMH}
In this first set of experiments, we focus on the segmentation of white matter lesions and six tissue classes (white matter, grey matter, basal ganglia, ventricles, cerebellum, brainstem), 
as well as the background.
%the white matter lesions and the background.
%
As detailed in Table~\ref{tab:DataCharacteristics}, we used 2 \emph{control} datasets and 2 \emph{lesion} datasets:

\begin{itemize}
    %\item
    \item\textbf{Lesion data} $\mathbf{S_{lesion}}$: The White Matter Hyperintensities (WMH) training database \citep{8669968} consists of 60 sets of brain MR images (\Tone and FLAIR, $M=2$) with manual annotations of WMHs. The data comes from three different institutes. Note that images modalities are be co-registered and resampled in the FLAIR coordinate space.
    
\item \textbf{Tissue data} $\mathbf{S_{control}}$: Consists of 35 \Tone scans ($M^{'}=1)$ from the OASIS project \citep{doi:10.1162/jocn.2007.19.9.1498} with annotations of 143 structures of the brain provided by Neuromorphometrics, Inc. (\url{http://Neuromorphometrics.com/}) under academic subscription. From the 143 structures, we deduct the 6 tissue classes. In order to have balanced training datasets between the two datasets, to include data acquired at the same field strength (3T) as the \emph{lesion} data, and similar to \citet{Li:IPMI:2017}, we added 25 \Tone control scans from the Alzheimer's Disease Neuroimaging Initiative 2 (ADNI-2) database  (\cite{doi:10.1002/jmri.21049}, \url{adni.loni.usc.edu}) with bronze standard parcellation of the brain structures computed with the accurate but time-consuming algorithm of \citet{Cardoso:TMI:2015}.

\item \textbf{Fully annotated data $\mathbf{S_{fully}}$:} MRBrainS18 (\url{http://mrbrains18.isi.uu.nl/}) is composed of 30 sets of brain MR images with tissue and lesion manual annotations. 7 scans are publicly available for training and validation. Although the cerebrospinal fluid (CSF) has been annotated in MRBrainS18, it was considered as background to have the same set of tissue classes as in $\mathbf{S_{control}}$ where the CSF was not labelled. Note that image modalities are be co-registered and resampled in the FLAIR coordinate space.
\end{itemize}

\begin{table*}[tb]
	\centering
	\caption{Evaluation of our framework (jSTABL) on patients with White Matter Lesion in comparison with baseline methods.
	We report means and standard deviations for Dice scores.  Means only are reported in the online leader-board, leading to missing standard deviations.
	}\label{tab:WhiteMatterLesion:Dice}
	\resizebox{\textwidth}{!}{
	\begin{tabular}{c *{12}{c}}
		\toprule
		\multirow{2}{*}{\bf Classes} & \multicolumn{3}{c}{\bf OASIS1+ADNI2} & \multicolumn{3}{c}{\bf WMH} & \multicolumn{4}{c}{\bf MRBrainS18}\\
		               \cmidrule(lr){2-4} \cmidrule(lr){5-7} \cmidrule(lr){8-11}
		               & Tissue & Fully-Sup & jSTABL & Pipeline & Fully-Sup & jSTABL & SPM & Pipeline & Fully-Sup & jSTABL  \\
		\midrule
		Grey matter &88.3 (3.4) &81.6 (2.5) &88.3 (3.4) &88.3 (2.1) &85.4 (2.7) & \textbf{88.8 (2.1)} & 76.5 & 82.3 & 83.7 & 82.2 \\
        White mater &92.8 (2.3) &83.3 (2.6) &92.3 (2.7) &92.1 (1.8) &85.4 (2.6) &\textbf{92.4 (1.5)} & 75.7 & 85.0 & 85.7 & 85.6  \\
        Brainstem &93.5 (1.0) &71.7 (2.7) &93.0 (0.9) &93.6 (1.0) &77.1 (2.4) &\textbf{94.2 (0.9)} & 76.5 & 72.8 & 85.0 & 73.3  \\
        Basal ganglia &89.5 (3.0) &69.6 (4.1) &88.4 (2.7) &\textbf{86.3 (4.3)} &74.2 (2.0) &85.1 (3.4) & 74.7 & 77.4 & 79.7& 78.0 \\
        Ventricles &90.3 (4.3) &70.5 (18.0) &90.6 (3.8) &94.7 (2.3) &92.1 (4.2) &\textbf{95.7 (1.4)} & 80.9 & 91.8 & 92.2 &  92.9& \\
        Cerebellum &95.0 (1.2) &92.0 (1.4) &94.9 (1.1) &95.7 (1.0) &93.8 (2.0) &\textbf{96.0 (0.9)}  & 89.4 & 89.2 & 93.2 & 90.4 \\
        \cmidrule(lr){1-1}
        White matter Lesion & & & &77.4 (9.6) &60.1 (19.1) &\textbf{77.6 (9.2)} & 40.8 & 58.4 & 56.2
        & 59.4 \\ 
		\bottomrule
	\end{tabular}
	}
	%\bigskip
\end{table*}
\begin{table*}[tb]
	\centering
	\caption{Evaluation of our framework (jSTABL) on patients with White Matter Lesion in comparison with baseline methods.
	We report means and standard deviations for 95th-percentile Hausdorff distances. Means only are reported in the online leader-board, leading to missing standard deviations.
	}\label{tab:WhiteMatterLesion:Haus}
	\resizebox{\textwidth}{!}{
	\begin{tabular}{c *{12}{c}}
		\toprule
		\multirow{2}{*}{\bf Classes} & \multicolumn{3}{c}{\bf OASIS1+ADNI2} & \multicolumn{3}{c}{\bf WMH} & \multicolumn{4}{c}{\bf MRBrainS18}\\
		               \cmidrule(lr){2-4} \cmidrule(lr){5-7} \cmidrule(lr){8-10}
		               & Tissue & Fully-Sup & jSTABL & Pipeline & Fully-Sup & jSTABL  & SPM & Pipeline & Fully-Sup & jSTABL  \\
		\midrule
        Grey matter &1.3 (0.4) &2.0 (0.4) &1.4 (0.5) &1.2 (0.2) &1.3 (0.3) &\textbf{1.1 (0.2)} &  2.9 & 1.9 & 1.9 & 2.1 \\
        White mater &1.2 (0.5) &3.0 (0.1) &1.2 (0.5) &\textbf{1.1 (0.2)} &2.0 (0.3) &\textbf{1.1 (0.1)} & 4.9 & 3.2& 2.9 & 3.2 \\
        Brainstem &1.4 (0.4) &10.5 (2.3) &1.7 (0.4) &1.7 (0.4) &8.8 (2.2) &\textbf{1.3 (0.4)} & 25.3 & 11.6 & 6.65 & 11.6 \\
        Basal ganglia &1.8 (0.5) &4.9 (0.8) &2.1 (0.3) &\textbf{2.0 (0.6)} &3.6 (0.4) &2.5 (0.3) & 7.1 & 4.3 & 4.3 & 4.0\\
        Ventricles &1.9 (2.5) &20.4 (12.2) &1.8 (2.5) &1.4 (1.2) &5.1 (10.5) &\textbf{1.0 (0.1)} & 5.8 & 3.2 & 3.0 & 2.9 \\
        Cerebellum &2.3 (0.6) &3.3 (0.4) &2.4 (0.6) &1.8 (0.7) &3.2 (1.7) &\textbf{1.8 (0.6)} & 4.3 & 5.1 & 3.7 & 4.8\\
        White matter Lesion & & & &4.6 (3.8) &11.0 (9.4) &\textbf{4.2 (3.6)} & 25.3 & 10.2 &13.3 & 7.2 \\
		\bottomrule
	\end{tabular}
	}
\end{table*}

\subsubsection{Similar acquisition protocol for the \texorpdfstring{\Tone}{T1} scans }
Despite the differences in scanners, all \Tone acquisitions across the datasets followed very similar protocols (MP-RAGE) (see Table~\ref{tab:DataCharacteristics}). Therefore, they were considered as following a similar distribution and the data was only pre-processed as follows.

\begin{itemize}
    \item \textbf{Skull stripping:} All the scans were skull-stripped using ROBEX \citep{ROBEX}.
    
    \item \textbf{Resampling:} All the scans in  $\mathbf{S_{control}}$ are resampled into the transversal direction with  slices of 3 mm thickness to obtain a similar spacing $1 \times 1 \times3$ mm\textsuperscript{3} in the datasets.
    
    \item \noindent\textbf{Intensity normalisation:} We used a zero-mean unit-variance normalisation in order to match the intensity distributions.
\end{itemize}

\subsubsection{Description of the compared models}
We considered three different models in our experiments.

\begin{itemize}
    % \item \textbf{Tissue segmentation model (\emph{Tissue})}: This task-specific model only performs tissue segmentation and is trained using the \Tone scans from the dataset with tissue annotations \textbf{S\textsubscript{control}}.
    % \item \textbf{Lesion segmentation model (\emph{Lesion})}: This task-specific model only performs lesion segmentation and is trained using the \Tone and FLAIR scans from the dataset with lesion annotations \textbf{S\textsubscript{lesion}}.
    \item \textbf{Pipeline model (\emph{Pipeline}):} This model corresponds to the combination of two task-specific models: 
    
    A \emph{Tissue} segmentation model that only performs tissue segmentation and is trained on the \Tone scans from the dataset with tissue annotations \textbf{S\textsubscript{control}}.
    
    A \emph{Lesion} model that only performs lesion segmentation and is trained using the \Tone and FLAIR scans from the dataset with lesion annotations \textbf{S\textsubscript{lesion}}.
    
    The two models are combined such that the predicted lesion mask has the priority over the predicted tissue mask. Consequently, the background of the \emph{Lesion} output is replaced by the \emph{Tissue} output.
    
    \item \textbf{Fully-supervised model (\emph{Fully-Sup}):} This joint model performs tissue and lesion segmentation and is trained using the \Tone and FLAIR scans from the small fully-annotated dataset  \textbf{S\textsubscript{fully}}.
    
    \item \textbf{Proposed joint model (\emph{jSTABL}):} Our proposed model for joint Segmentation of Tissues and Brain Lesions 
    %performs tissue and lesion segmentation and
    is trained using both the \Tone scans from \textbf{S\textsubscript{control}} with the tissue annotations and the \Tone and FLAIR scans from \textbf{S\textsubscript{lesion}} with the lesion annotations. 
    %The training procedure follows our method.
\end{itemize}

Each model used the architecture presented in Fig.~\ref{fig:network}. Consequently, the \emph{Pipeline} model has twice as many parameters as the other models. 

In this set of experiments, the skull-stripped images are first cropped to remove the blank spaces and then padded to size of $(144,192,48)$.

\subsubsection{Method for assessing the models}
The performance of the three models was evaluated on the three datasets using Dice Score and $95\%$ Hausdorff distance. On the \emph{control} data (OASIS1+ADNI2) and WMH, scores were computed on the testing splits, while on MRBrainS18, models were submitted to the challenge MRBrainS18.

For the \emph{control} data (OASIS1+ADNI2) and MRBrainS18, the full set of annotations allows a direct assessment of the tissue and the lesion segmentation performance. For WMH, only the \emph{lesion} annotations are  provided. In order to assess both the tissue and lesion segmentation on WMH, the lesions are filled as normal-appearing white matter on \Tone images using the method described in \citet{PRADOS2016376} and implemented in NiftySeg \citep{Cardoso:TMI:2015}. Then, GIF framework \citep{Cardoso:TMI:2015} was performed on the modified \Tone scans to obtain bronze standard tissue annotations. The tissue mask and lesion annotations were then merged by completing the non-lesion parts with the tissue mask. Finally, the model outputs are compared to the merged tissue and lesion masks. In the end, for each model and each dataset, we can assess the performance of tissue and lesion segmentation.

Given that participants to the MRBrainS18 challenge do not have access to the held-out evaluation data set and that the Jaccard score is not provided by the challenge organisers, only the Dice Similarity Coefficient (DSC) and 95th-percentile Hausdorff distance are reported for each class.

\subsubsection{Results}
The main results are shown in Table~\ref{tab:WhiteMatterLesion:Dice} for the Dice Similarly Coefficient and in Table~\ref{tab:WhiteMatterLesion:Haus} for the 95th-percentile Hausdorff distance. 

% Firstly, our proposed method (\emph{jSTABL}) achieves comparable performance to the single-task models on Neuromorphometrics (\emph{Tissue}).and on WMH (priority of \emph{Lesion} in \emph{Pipeline}). This suggests that learning from hetero-modal datasets via our method does not degrade the task-specific performance.

Firstly, our proposed method (\emph{jSTABL}) achieves comparable performance to the single-task models on the \emph{control} data (\emph{Tissue}) and on WMH (priority of \emph{Lesion} in \emph{Pipeline}). This suggests that learning from hetero-modal datasets via our method does not degrade the performance on the tasks characterising the task-specific datasets.

Secondly, \emph{jSTABL} slightly outperforms \emph{Pipeline} on segmenting the tissues in WMH for the two sets of metrics. This shows that the tissue knowledge learnt from \Tone scans has been well generalised to multi-modal scans. Although we could have expected that the presence of lesions would create perturbations for the \emph{Tissue} model, this latter model in fact ignores the lesions and mostly classifies them as white matter. Given that the white matter lesions are usually surrounded by white matter, the \emph{Pipeline} predictions are consequently not too degraded. However, some artefacts around the lesions in the \emph{Pipeline} outputs can be observed, in particular in the ventricles for patients with large lesions surrounding them. Fig.~\ref{fig:resultsWMH}(a) shows an example for which parts of the ventricles are classified as background. In contrast, we did not observe such artefacts with \emph{jSTABL} predictions.

\begin{figure}[tbp]
  \centering
  \includegraphics[width=\linewidth]{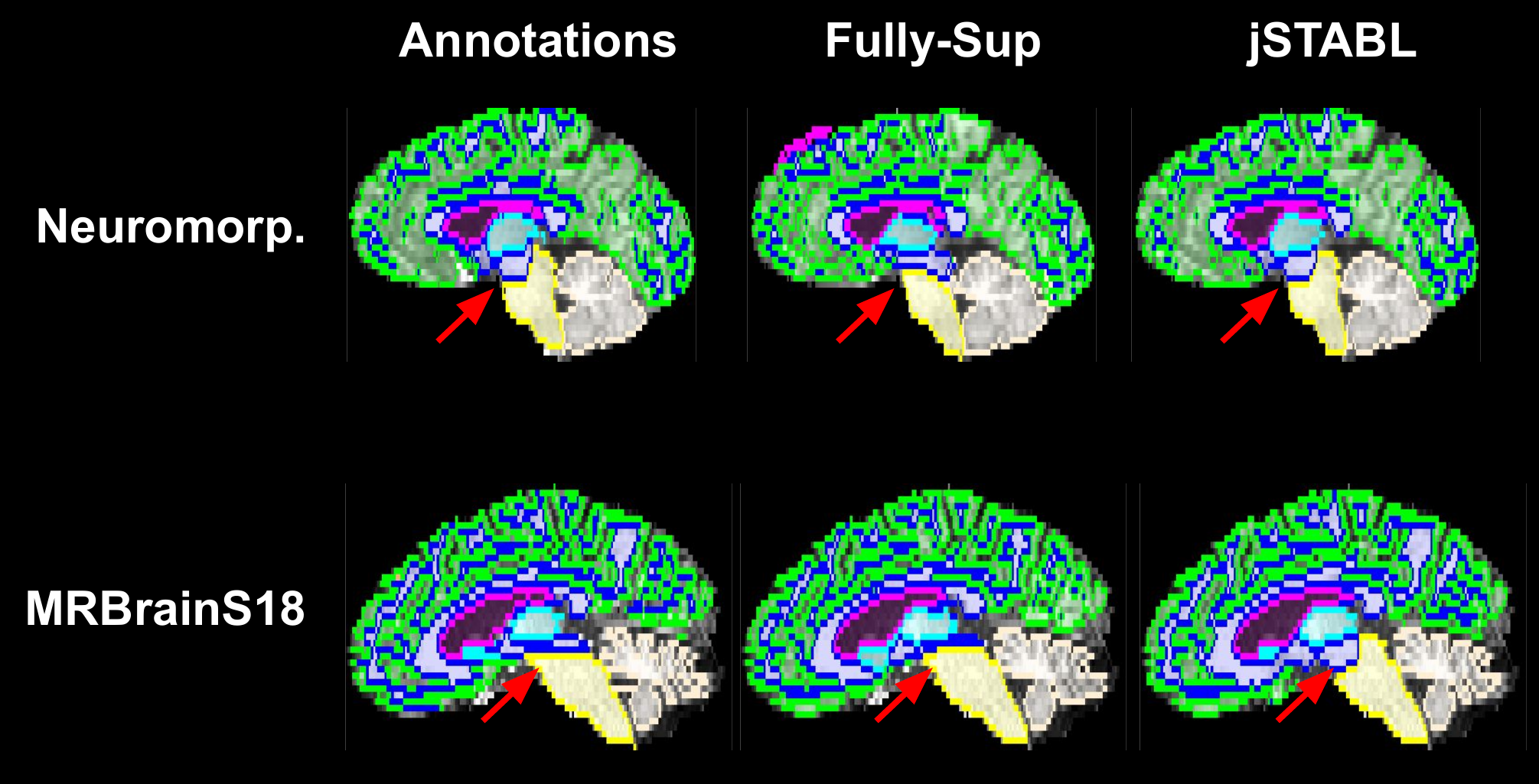}
  \caption{Comparison of the brainstem annotations by Neuromorphometrics and in MRBrainS18 and between the outputs of the \emph{Fully-Sup} and \emph{jSTABL} models. Arrows show the annotations protocol differences.}
  \label{fig:differences}
\end{figure}

% Thirdly, the comparison between the fully-supervised model \emph{Fully-Sup} trained on a small dataset and \emph{jSTABL} demonstrates the main advantage of using large existing datasets: Unlike the fully-supervised model (\emph{Fully-Sup}), \emph{jSTABL} generalises well on unseen data. While the fully-supervised model (\emph{Fully-Sup}) dramatically fails to segment scans from Neuromorphometrics and WMH, the \emph{jSTABL} model obtains relatively good performance on the three datasets for the tissue and lesion segmentation. In particular, \emph{jSTABL} outperforms SPM on 6 of the 7 classes. In fact, the only class that is significantly underperformed compared to the fully-supervised model (\emph{Fully-Sup}) is the brain stem. This is due to some differences in the protocol of annotations across Neurophometrics and MRBrainS datasets. Figure~\ref{fig:differences} shows these differences and the consequences on the prediction.

Thirdly, \emph{jSTABL} outperforms the fully-supervised model (\emph{Fully-Sup}) on the \emph{control} data (OASIS1+ADNI2) and WMH, while reaching comparable performance on MRBrainS18. This demonstrates the two main advantages of our method. First, without using any fully-annotated data, our model performs as well as a fully-supervised model that could be considered as an upper bound for our method, especially when the testing and training splits are from the same dataset (MRBrainS18). Secondly, our method takes advantage of large task-specific datasets: Unlike the fully-supervised model (\emph{Fully-Sup}), \emph{jSTABL} generalises well on unseen data (MRBrainS18). While the fully-supervised model (\emph{Fully-Sup}) fails to segment scans from OASIS1, ADNI2 and WMH, the \emph{jSTABL} model obtains relatively good performance on all the datasets  we use for tissue and lesion segmentation. In particular, \emph{jSTABL} outperforms SPM on 6 of the 7 classes. In fact, the only class that is significantly underperformed compared to the fully-supervised model (\emph{Fully-Sup}) is the brain stem. This is due to observed differences in the annotation protocol across the \emph{control} and MRBrainS datasets. Figure~\ref{fig:differences} shows these differences and the consequences on the prediction.

\begin{figure*}[tb!]
 % Caption and label go in the first argument and the figure contents
 % go in the second argument
  \centering
  \includegraphics[width=\linewidth]{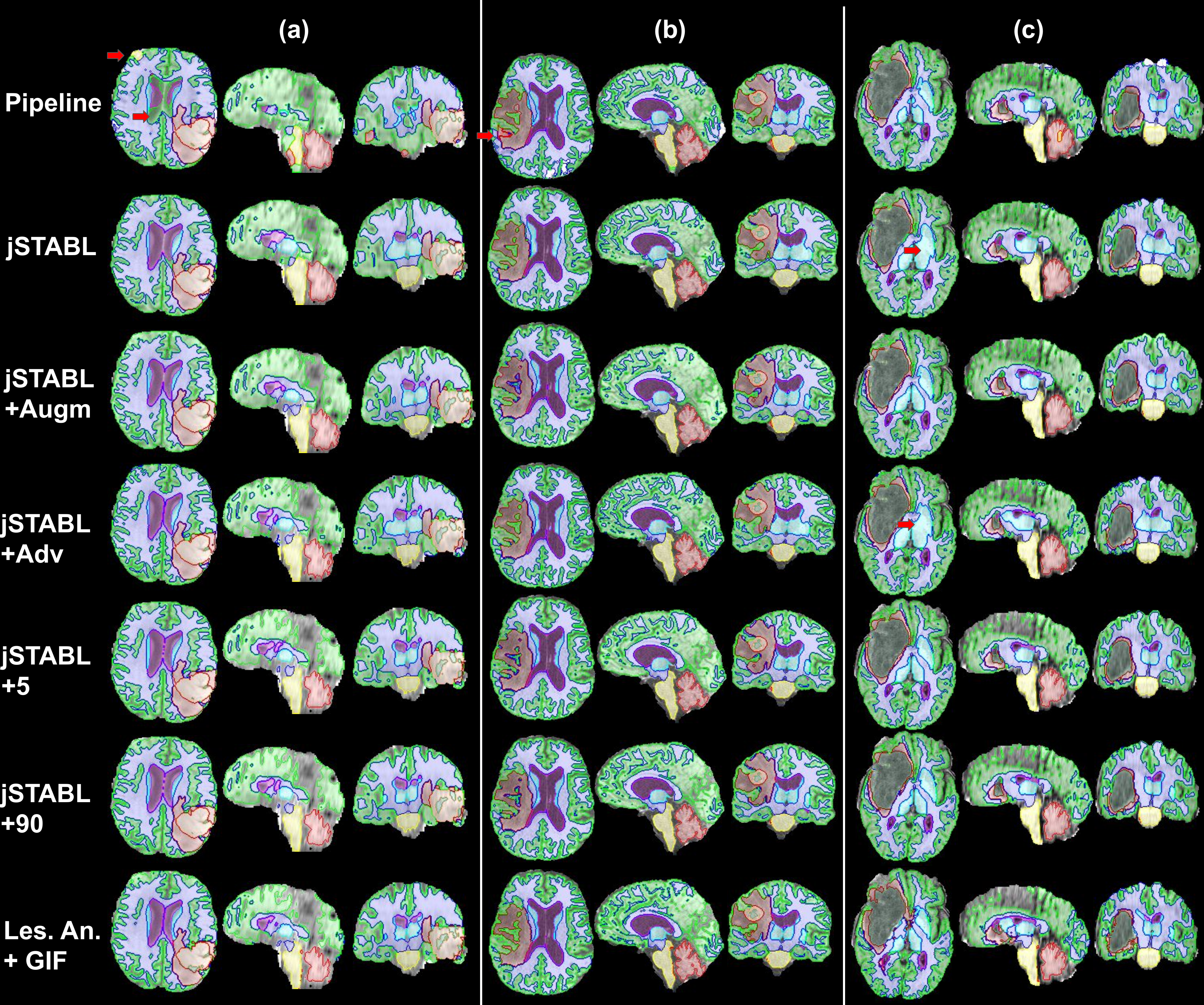}
  \caption{Examples of multi-modal outputs shown on \Tone scans from BraTS18 for the different models: \emph{Pipeline}, \emph{jSTABL}, \emph{jSTABL+Augm}, \emph{jSTABL+Adv}, \emph{jSTABL+5}, \emph{jSTABL+90}  and \emph{Les. Ann. + GIF} pipeline. Scans with different resolutions, contrasts and grades (High Grade for (a) and (b), Low Grade for (c)) are presented.}
  \label{fig:BraTS_prediction}
\end{figure*}

%%%
\subsection{Glioma and tissue segmentation}
\subsubsection{Task and datasets}
Additionally, we assess our framework on another main types of brain lesions: Gliomas. Our goal is to segment the 6 tissue classes and three tumour classes (whole tumour, core tumour, enhancing tumour). In this case, domain adaptation was required and its evaluation is the focus of this section.
%which allows us to specifically evaluate the performance of this component.
%
We used two sets of data in these experiments.

\begin{itemize}
    \item \textbf{Tissue data $\mathbf{S_{control}}$:} again we used OASIS1 data and the same 25 \Tone \emph{control} scans from ADNI2 with tissue annotations as presented in section \ref{experiments:WMH}.
    
    \item \textbf{Lesion data $\mathbf{S_{lesion}}$:} We evaluate our method on the training set of BraTS18 \citep{BraTS,DBLP:journals/corr/abs-1811-02629} which contains the scans of 285 patients, 210 with high grade glioma and 75 with low grade glioma. 129 patients have a tumour located in one hemisphere only. Four scans (\Tone, \Tonec, \Ttwo and FLAIR) have been acquired for each patient and pre-processed by the organisers: Co-registration, skull-stripping and re-sampling to an isotropic 1mm resolution. Manual annotations include three tumour labels: 1) Necrotic core and non-enhancing tumour; 2) oedema; and 3) enhancing core. 
\end{itemize}

The acquisition protocols of the \Tone scans in the two datasets are inconsistent. Specifically, MP-RAGE was used for the tissue data $\mathbf{S_{control}}$, while we observed other protocols such as fast spin echo (SE) for $\mathbf{S_{lesion}}$. Note that the detailled acquisition settings for $\mathbf{S_{lesion}}$ are not publicly available.

\subsubsection{Description of the compared models}
In order to evaluate our framework with and without the domain adaptation (DA) component, different models are considered, as presented in \ref{section:da}.

\begin{itemize}
    \item \textbf{Pipeline model (\emph{Pipeline})}:  This model corresponds to the combination of two task-specific models:     
    
    A \emph{Tissue} segmentation model that only performs tissue segmentation and is trained on the \Tone scans from the dataset with tissue annotations \textbf{S\textsubscript{control}}.
    
    A \emph{Lesion} model that only performs lesion segmentation and is trained using the \Tone, \Tonec, \Ttwo and FLAIR scans from the dataset with lesion annotations \textbf{S\textsubscript{lesion}}.

    \item \textbf{Proposed joint model without DA (\emph{jSTABL})}: Our joint Segmentation Tissue And Brain Lesion model is trained using our training procedure without domain adaptation, tissue segmentation is learned from the \Tone scans in $\mathbf{S_{control}}$.
    
    \item \textbf{jSTABL + data augmentation (\emph{jSTABL+Augm})}: Corresponds to our \emph{jSTABL} model with DA based on data augmentation.

    \item \textbf{jSTABL + adversarial DA (\emph{jSTABL+Adv})}: Corresponds to our \emph{jSTABL} model with DA based on adversarial learning.
    
    \item \textbf{jSTABL + 5 synthetic control scans (\emph{jSTABL+5})}: Corresponds to our \emph{jSTABL} model with only $5$ additional pseudo-healthy scans. 
    
    \item \textbf{jSTABL + 90 synthetic control scans  (\emph{jSTABL+90})}: Corresponds to our \emph{jSTABL} model with $90$ additional pseudo-healthy scans. 
\end{itemize}

Note that the pseudo-healthy scans used for training were generated from the training \emph{lesion} scans to avoid introducing bias at testing stage. 
 
In this set of experiments, the skull-stripped images are first cropped to remove the blank spaces and then random patches of size $(112,112,112)$ are fed to the network.

% \begin{table*}[htb!]
% 	\centering
% 	\caption{Evaluation of our framework (jSTABL) on ADNI1
% 	We report means and standard deviations for Dice scores. Metric were computed on the BraTS 2018 validation dataset.}
% 	%\resizebox{\textwidth}{!}{
% 	\begin{adjustbox}{width=1\textwidth}
% 	\begin{tabular}{cccccccc}
% 		\toprule
% 		\multirow{2}{*}{\bf Models} & \multicolumn{2}{c}{\bf no-DA} & \multicolumn{3}{c}{\bf Unsupervised-DA} & \multicolumn{2}{c}{\bf Supervised} \\
% 		%
% 		                \cmidrule(lr){2-3} \cmidrule(lr){4-6}  \cmidrule(lr){7-8}
% 		               &    Pipeline	    &     jSTABL    & jSTABL + Adv &  jSTABL + PC   &      jSTABL + Augm	    &     jSTABL + 5    & jSTABL + 90    \\
% 		\midrule
%         Grey Matter &93.0 (1.9) &92.6 (1.8) &92.7 (1.9) &90.7 (2.2) &92.4 (1.5) &92.4 (1.5) &92.7 (1.4) \\
%         %
%         White Matter &94.9 (1.5) &94.3 (1.9) &94.5 (1.8) &93.0 (2.4) &94.1 (1.3) &94.3 (1.8) &94.5 (1.7)  \\ 
%         %
%         Brainstem &95.1 (1.3) &95.0 (1.2) &95.1 (1.2) &94.2 (1.1) &95.0 (1.2) &94.8 (1.2) &94.8 (1.3) \\ 
%         %
%         Basal ganglia &91.3 (2.3) &91.0 (2.1) &91.2 (2.2) &88.9 (2.3) &91.5 (2.2) &91.5 (2.2) &91.5 (2.4) \\ 
%         %
%         Ventricules &92.8 (3.8) &92.6 (4.0) &92.6 (4.1) &91.4 (4.6) &92.8 (3.4) &92.7 (3.8) &92.6 (4.0)  \\ 
%         %
%         Cerebellum &96.4 (1.3) &96.2 (1.2) &96.5 (0.9) &95.4 (1.3) &96.7 (0.6) &96.4 (0.9) &96.5 (0.8)\\
%         %
% 		\bottomrule
% 	\end{tabular}
% 	%}
% 	\end{adjustbox}
% 	\label{tab:BraTS_results}
% \end{table*}

\begin{table*}[htb!]
	\centering
	\caption{Evaluation of our framework (jSTABL) on patients with gliomas in comparison to baseline methods.
	We report means and standard deviations for Dice scores. Metrics were computed on the BraTS 2018 validation dataset.}
	%\resizebox{\textwidth}{!}{
	\begin{adjustbox}{width=1\textwidth}
	\begin{tabular}{ccccccc}
		\toprule
		\multirow{2}{*}{\bf Models} & \multicolumn{2}{c}{\bf w/o DA} & \multicolumn{2}{c}{\bf w/ DA} & \multicolumn{2}{c}{\bf w/ pseudo-healthy gen.} \\ \cmidrule(lr){2-3} \cmidrule(lr){4-5}  \cmidrule(lr){6-7}  
		               &    Pipeline	    &     jSTABL    & jSTABL + Adv    &      jSTABL + Augm	    &     jSTABL + 5    & jSTABL + 90    \\
		\midrule
        Grey Matter &76.1 (8.4) &79.1 (5.2) &81.1 (4.6) &82.8 (4.7) &88.3 (3.9) &88.8 (4.0) \\
        White Matter &85.4 (5.7) &87.0 (4.4) &88.1 (4.4) &90.3 (2.8) &93.1 (2.5) &93.3 (2.6) \\
        Brainstem &81.5 (17.6) &92.4 (2.5) &92.6 (1.9) &92.4 (2.0) &94.9 (1.4) &95.5 (1.5) \\
        Basal Ganglia &72.7 (20.1) &73.1 (7.3) &77.7 (7.1) &84.7 (5.1) &89.7 (3.5) &90.5 (3.2) \\
        Ventricles &75.0 (26.9) &91.8 (6.4) &92.5 (4.5) &93.4 (4.4) &94.7 (4.2) &95.1 (3.8) \\
        Cerebellum &86.5 (11.3) &93.2 (4.6) &93.7 (3.5) &94.0 (2.4) &94.7 (4.5) &95.1 (2.9) \\
        \cmidrule(lr){1-1}
        Whole Tumour &87.9 (8.7) &88.1 (6.7) &87.7 (8.1) &88.3 (9.0) &88.2 (8.1) &88.1 (9.4) \\
        Core Tumour &78.6 (20.5) &79.1 (19.6) &79.5 (18.9) &80.4 (18.6) &80.5 (18.1) &80.9 (17.9) \\
        Enhancing Tumour &69.9 (29.1) &70.0 (28.6) &70.0 (28.9) &71.4 (27.7) &70.3 (28.8) &71.0 (28.3) \\
		\bottomrule
	\end{tabular}
	%}
	\end{adjustbox}
	\label{tab:BraTS_results}
\end{table*}

\subsubsection{Method for assessing the models}
While the evaluation of the tumour segmentation on BraTS18 and the tissue segmentation on the \emph{control} data (OASIS1+ADNI2) is straightforward using the manual annotations, the tissue segmentation performance cannot be assessed on BraTS18 due to the missing tissue annotations. For this reason we propose two methods to assess quantitatively and qualitatively the tissue segmentation on BraTS18.

\paragraph{Quantitative assessment using the symmetrised data}
Firstly, we propose to use the 129 patients from BraTS18 with a tumour located in one side to generate 129 pseudo-healthy symmetrised data with the bronze standard tissue annotations from GIF as ground truth. Examples are shown in Figure~\ref{fig:symetrized_brains}. By computing the Dice Score Coefficient between the predictions on the symmetrised BraTS18 data and the bronze standard ground truth, we quantitatively evaluate our model on the pseudo-healthy hemisphere of BraTS18 samples. 

\paragraph{Qualitative assessment on anatomical landmarks}
Secondly, in order to assess the accuracy of the models on the tissues surrounding the tumour, we propose a new qualitative protocol. This protocol is based on the Alberta Stroke Program Early CT Score (ASPECTS) which was originally proposed to assess early ischaemic cerebral changes on CT or MRI scans \citep{BARBER20001670}. The stroke scores are obtained by assessing the integrity of 10 anatomical landmarks as shown on \ref{fig:ASPECTS}. Scores and associated template are commonly used in clinical practice.

\begin{figure*}[!tbp]
  \centering
  \subfloat[Anatomy ASPECTS+ - mostly/highly accurate]{\includegraphics[width=0.5\textwidth]{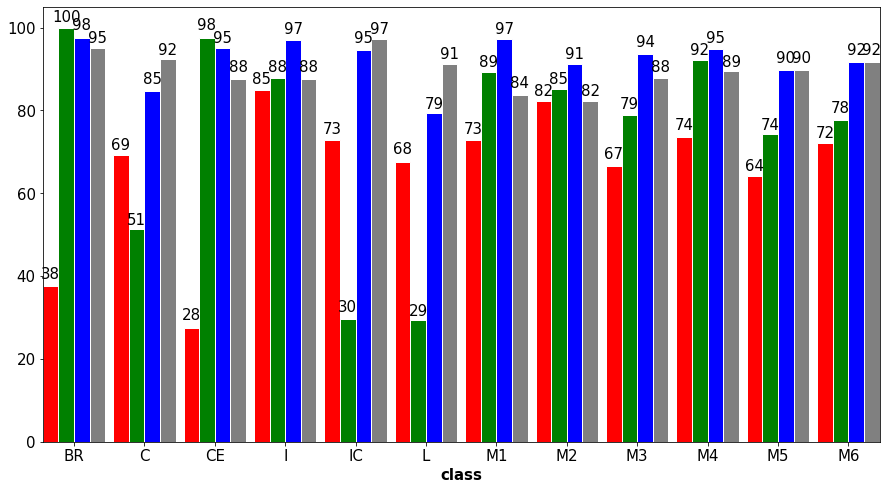}\label{fig:f1}}
  \hfill
  \subfloat[Anatomy ASPECTS+ - highly accurate]{\includegraphics[width=0.5\textwidth]{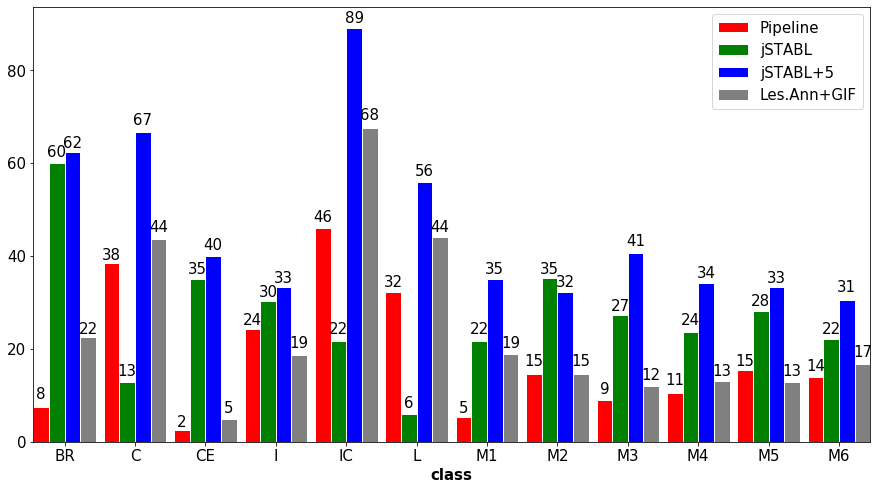}\label{fig:f2}}
  \caption{\textbf{Comparison of our method (jSTABL) with the GIF framework using the proposed Anatomy ASPECTS+ qualitative assessment methodology.} BR = Brainstem, C = Caudate, CE = Cerebellum, I = Insula, IC = Internal Capsule,  L = Lentiform Nucleus, M1 = Frontal operculum, M2 = Anterior temporal lobe, M3 = Posterior temporal lobe, M4 = Anterior MCA, M5 = Lateral MCA, M6 = Posterior MCA}
  \label{fig:ASPECTS_score}
\end{figure*}

\begin{figure}[bt!]
 % Caption and label go in the first argument and the figure contents
 % go in the second argument
  \includegraphics[width=\linewidth]{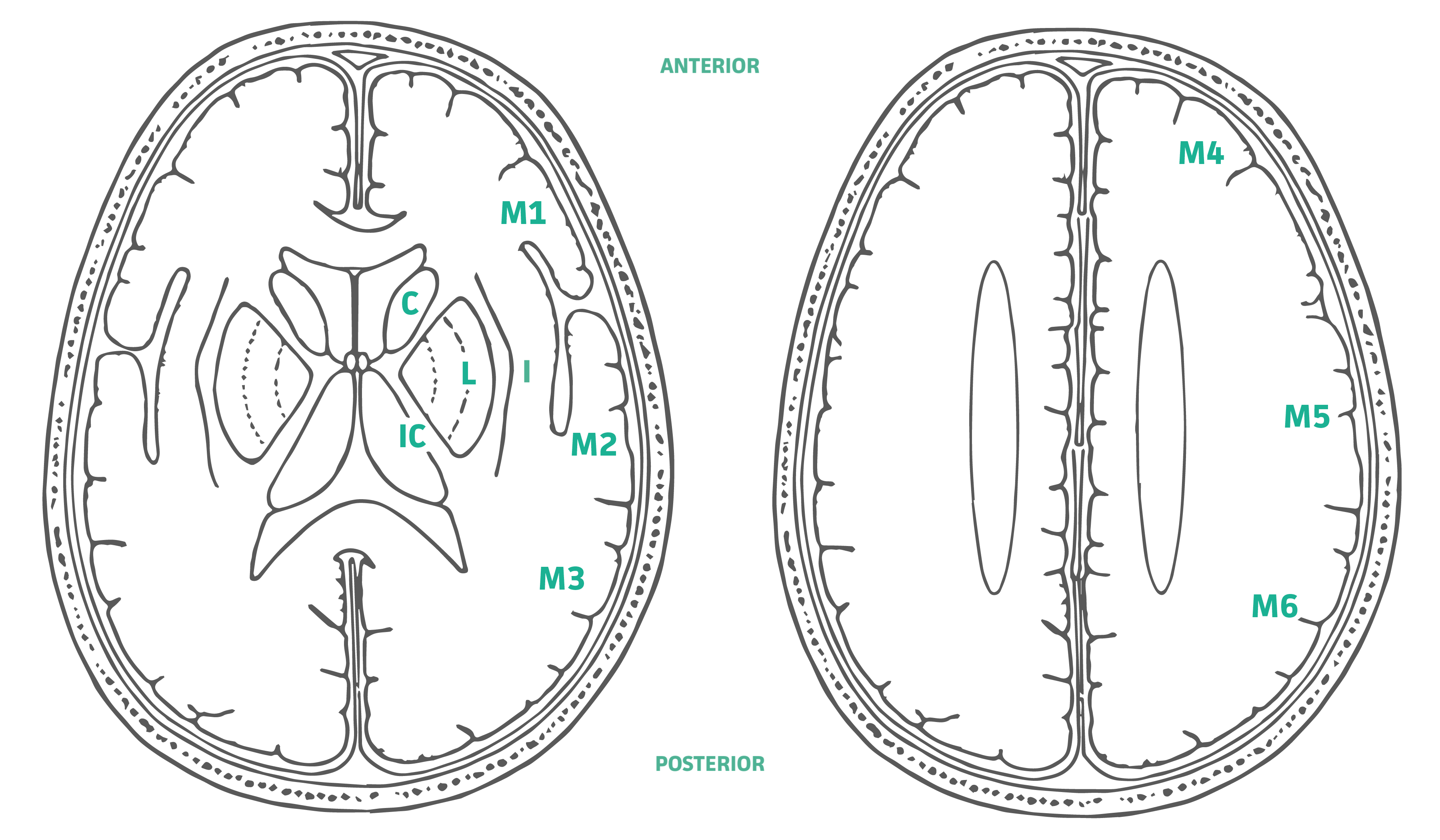}
  \caption{ASPECTS anatomical landmarks \citep{BARBER20001670} used in our qualitative assessment methodology in the absence of joint ground truth. C = Caudate, I = Insula, IC = Internal Capsule,  L = Lentiform Nucleus, M1 = Frontal operculum, M2 = Anterior temporal lobe, M3 = Posterior temporal lobe, M4 = Anterior MCA, M5 = Lateral MCA, M6 = Posterior MCA Illustration courtesy of P.A. Barber.}
  \label{fig:ASPECTS}
\end{figure}

The landmarks were chosen because they are easily identifiable, reliable amongst readers and capture a large cerebral coverage. The landmarks include or delineate our tissue classes of interest: Grey matter; white matter; basal ganglia; and ventricles. Instead of evaluating loss of clarity of landmarks due to ischemia we evaluated loss of clarity of landmarks due to incorrect tissue predictions. Unlike ASPECTS, which excludes infratentorial structures which are difficult to evaluate on CT, we added the brainstem and the cerebellum as two additional landmarks for a total of 12 anatomical landmarks. We named our assessment method Anatomy ASPECTS+. For each landmark, 3 scores are possible: 0 = anatomy inaccurate; 0.5 = anatomy mostly accurate; and 1 = anatomy highly accurate. Anatomical landmarks that were infiltrated with substantial tumour were excluded. 

For our experiments, we randomly drew 20 patients from the testing sets of BraTS18 and two senior neuro-radiologists independently evaluated the quality of the predictions using Anatomy ASPECTS+ for 4 methods: 1) \emph{Les. Ann. + GIF} pipeline that uses the tumour annotations and tissue segmentation obtained by GIF while the tumour is masked; 2) \emph{Pipeline}; 2) \emph{jSTABL}; 3)  \emph{jSTABL+5};  4)  \emph{jSTABL+10}.
The neuro-radiologists assessed blindly the 4 methods, in a randomised order, for the 20 patients.

\subsubsection{Results}
Table~\ref{tab:BraTS_results} shows the DSC for the 6 tissue classes and the three tumour classes on BraTs18.

Firstly, \emph{jSTABL} model outperforms \emph{Pipeline} on tissue segmentation. We observed that the presence of a large tumour creates major perturbations for the \emph{Tissue} model. For example, we found samples for which the tumour and the surrounding tissues were partially classified as cerebellum, even though the tumour was far from the cerebellum, as shown in Fig.~\ref{fig:BraTS_prediction}. In contrast, such artefacts were not observed for \emph{jSTABL} model, demonstrating again advantages of our method compared to a simpler \emph{Pipeline} approach.

Secondly, while obtaining relatively good performance on most of the tissue classes, \emph{jSTABL} model fails to segment correctly grey matter and basal ganglial. This highlights the needs for domain adaptation.

Thirdly, learning from pseudo-healthy annotated scans (\emph{jSTABL+5} and  \emph{jSTABL+90}) outperforms the other unsupervised DA strategies based either on data augmentation (\emph{jSTABL+Augm}) and adversarial learning (\emph{jSTABL+Adv}). This demonstrates the benefits of using a supervised approach for our problem. Moreover, only 5 pseudo-healthy annotated scans are required to obtain an accuracy similar to the one on the \emph{control} data (see Table~~\ref{tab:WhiteMatterLesion:Dice}), i.e. to bridge the domain gap. Figure~\ref{fig:BraTS_prediction} shows that learning from pseudo-healthy annotated scans allows the network to be robust to variations in resolution, contrast or glioma grade, even with few samples used for domain adaptation.

Finally, Anatomy ASPECTS+ is employed to provide a quantitative assessment of the segmentation of the tissues surrounding the tumour. Four models are compared by two neuro-radiologists: \emph{Pipeline}, \emph{jSTABL}, \emph{jSTABL+$5$} and the time-consuming \emph{Les. Ann. + GIF} pipeline that requires manual annotations of the lesions. Results are presented in Fig~\ref{fig:ASPECTS_score}. First, \emph{jSTABL} is more often ``mostly accurate" than the \emph{Pipeline} (mean score - $75\%$ vs $66\%$). Again, this highlights the strength of our joint model compared to a pipeline approach. Secondly, \emph{jSTABL+$5$} is more often ``highly accurate" than the \emph{Les. Ann. + GIF} pipeline (mean score - $46\%$ vs $24\%$). This shows that our fast and fully automatic method can be considered as a new state-of-the-art for performing joint tissue and lesion segmentation.

% is mostly accurate in $75\%$ of the cases, while \emph{Pipeline} is only accurate in $66\%$ of the times.

% that all the methods, except \emph{Pipeline}, are mostly accurate 

% alghough

% we observe that \emph{jSTABL} outperforms \emph{Pipeline}

% In more than $90\%$ of the cases, our proposed algorithms were mostly or highly accurate. Although these results are lower than the \emph{Les. Ann. + GIF} pipeline, they are promising. In particular, while the \emph{Les. Ann. + GIF} pipeline is not scalable in practice, our method is fast and fully-automatic.

% \begin{figure}[tb!]
%  % Caption and label go in the first argument and the figure contents
%  % go in the second argument
% \floatconts
%   {fig:differences}
%   {\caption{Annotation protocol comparison between scans from (a) Neuromorphometrics and (b) MRBrainS18. (1) sagital slice from test images volumes, (2) manual annotations, (3) outputs from our method \emph{(W+N)}, (4) outputs from fully-supervised model \emph{(N)}. Arrows show the protocol differences.}}
%   {\includegraphics[width=\linewidth]{comparatif}}
% \end{figure} 

\section{Discussion}

In this section, we discuss some of the limitations of the different methods. 

Firstly, a common modality across the task-specific dataset is required to transfer the knowledge learn between the task-specific sets of modality. Without this common modality, the upper bound is not tractable anymore and our method cannot be applied.

Secondly, our approach relies on a simple hetero-modal architecture that aims to encode modalities in a common shared feature space. Yet, averaging the feature maps doesn't enforce the network to learn a shared feature representation. To tackle this problem, a hetero-modal variational auto-encoder architecture has been recently introduced \cite{Dorent:MICCAI:2019}. Based on a principled formulation of the problem, the induced loss function is the cross-entropy, which does not satisfy the triangle inequality. Consequently, without further research, this approach cannot be directly integrated in the formulation of our problem. 

Thirdly, we found that the presence of lesions does not always perturbed a network trained on \emph{control} data, especially for small lesions. Consequently, our method didn't always show large improvements compared to a simpler \textit{Pipeline} approach on WMH. However, we observed that \textit{Pipeline} can be perturbed by larger pathology and thus is less robust.
%However, we emphasise that, for risk management purposes, relying on an unpredictable behaviour should be avoided in clinical applications.

\section{Conclusion}
This work addresses the challenge of learning a joint brain tissue and lesion segmentation with disjoint heterogeneous annotations. Our novel approach is mathematically grounded, conceptually simple, and relies on reasonable assumptions. 

The main contribution of this work is to overcome the challenge of the lack of fully-annotated data for joint problems. We demonstrate that a model trained on databases providing either the tissue or the lesion annotations and with different modalities can achieve similar performance to a model trained on a fully-annotated joint dataset. Our work also shows that the knowledge learnt from one modality can be preserved when more modalities are used as input. Finally, domain adaptation for image segmentation can be performed with a small set of data related to the target distribution.

In the future, we will evaluate our approach on new datasets with other lesions.
Furthermore, we would like to extend our method to include the full parcellation of the brain (143 structures). Finally, we plan to integrate uncertainty measures in our framework as a future work.
As one of the first work to methodologically address the problem of joint learning from hetero-modal and domain-shifted datasets, we believe that our approach will help DNN make further impact in clinical scenarios.

\section*{Acknowledgements}%
This work was supported by the Wellcome Trust [203148/Z/16/Z] and the Engineering and Physical Sciences Research Council (EPSRC) [NS/A000049/1]. TV is supported by a Medtronic / Royal Academy of Engineering Research Chair [RCSRF1819\textbackslash7\textbackslash34].
CS is supported by an the Alzheimer's Society Junior Fellowship [AS-JF-17-011].
We thank Prof. Alexander Hammers, Samuel Joutard and Lucas Fidon for their useful comments.
Some of the data used in preparation of this article was obtained from the Alzheimer’s Disease Neuroimaging Initiative (ADNI) database (\url{adni.loni.usc.edu}). As such, the investigators within the ADNI contributed to the design and implementation of ADNI and/or provided data but did not participate in analysis or writing of this report. A complete listing of ADNI investigators can be found at \url{http://adni.loni.usc.edu/wp-content/uploads/how_to_apply/ADNI_Acknowledgement_List.pdf}

\bibliographystyle{model2-names.bst}\biboptions{authoryear}
\bibliography{refs}

%%%%%
\clearpage
\appendix
\phantomsection
\addcontentsline{toc}{section}{Appendices}%
%%%
\section{Probabilistic multi-class Jaccard loss function}
First we recall the definitions of the binary Jaccard distance and the proposed extension to probabilistic inputs.
The binary Jaccard distance $J_{bin}$ is defined such that:
\begin{equation}\label{eq:binary_jaccard_2}
    %\begin{split}
        \forall a,b \in \{0,1\}^{N}, \ J_{bin}(a,b) = 1- \frac{\sum_{i=1}^{N} a_i b_i}{\sum_{i=1}^{N} a_i + b_i-a_i b_i} 
    %\end{split}
\end{equation}%
\begin{definition} 
(Probabilistic multi-class Jaccard distance)\\
Let $C$ be the number of classes in $\mathcal{C}$, $N$ be the number of voxels and $\mathcal{P} \subset [0,1]^{C\times N} $ denote the set of probability vector map such that for any $p=(p_{c,i})_{c\in \mathcal{C}, i\in[0;N]} \in \mathcal{P}$:
    $$\forall i \in [0;N], \ \sum_{c\in C} p_{c,i}=1 $$
    The probabilistic multi-class Jaccard distance is defined for any $(u,v)\in\mathcal{P}^2$ as:
\begin{equation}\label{eq:jaccard_prob_2}
\begin{split}
    %\forall (u,v)\in \mathcal{P}^2, \
    \mathcal{J}(u,v) =\sum_{c \in C} \omega_{c} \underbrace{ \frac{2\sum_{i=1}^{N}|u_{c,i}-v_{c,i}|}{\sum_{i=1}^{N}|u_{c,i}|+|v_{c,i}|+|u_{c,i}-v_{c,i}|}}_{\mathcal{J}_{c}} 
\end{split}
\end{equation}
where $\omega_{c}$ are class-specific weights summing up to one.
\end{definition}

%%%
\subsection{Relation between the probabilistic Jaccard loss and the binary Jaccard distance}\label{appendix:relationship}
The binary case corresponds to a two-class problem, i.e. $\mathcal{C}=\{0,1\}$. Let $a,b\in \{0,1\}^{N}$ be two binary vectors of size $n$ and let $u$ and $v$ denote respectively the categorical encodings of $a$ and $b$, i.e. $a_i=1 \iff \left(u_{1,i}=1 \text{ and } u_{0,i}=0\right)$ and $b_i=1 \iff \left(v_{1,i}=1 \text{ and } v_{0,i}=0\right)$. 
The binary Jaccard distance can be rewritten as:
\begin{equation*}
    \mathcal{J}_{bin}(a,b) = 1- \frac{\sum_{i=1}^{N} a_i b_i}{\sum_{i=1}^{N} a_i + b_i-a_i b_i} = \frac{\sum_{i=1}^{N} a_i + b_i-2a_i b_i}{\sum_{i=1}^{N} a_i + b_i-a_i b_i}
\end{equation*}
Given that for all $i\in \{1,\dots,N\}$, $a_i=0$ or $a_i=1$, we observe that $a_i^2=a_i$. Using the same property for $b$, we get:
\begin{equation*}
    J_{bin}(a,b) = \frac{\sum_{i=1}^{N} (a_i - b_i)^2}{\sum_{i=1}^{N} a_i^2 + b_i^2-a_i b_i} = \frac{2\sum_{i=1}^{N} (a_i - b_i)^2}{\sum_{i=1}^{N} a_i^2 + b_i^2+(a_i - b_i)^2}
\end{equation*}
Finally, given that for all  $i\in \{1,\dots,N\}$,  $(a_i - b_i)^2=0$ or $(a_i - b_i)^2=1$, we have that $(a_i - b_i)^2=|a_i - b_i|$ and conclude that:
\begin{equation*}
    \mathcal{J}_{bin}(a,b) = \frac{2\sum_{i=1}^{N} |a_i - b_i|}{\sum_{i=1}^{N} a_i^2 + b_i^2+|a_i - b_i|} =  \mathcal{J}_{1}(u,v)
\end{equation*}

%%%
\subsection{Proof of Lemma \ref{lemma:probabilisticjaccard}}
\label{appendix:proofJaccarddistance}
The proof that the probabilistic Jaccard is a distance is based on the Steinhaus transform \citep{Spaeth:ORS:1981}. Given a metric space $(E,d)$ with a distance $d$ and given a fixed point $\alpha \in E$, we can define a new distance $d_{new}$ as:
$$ d_{new}(x,y) = \frac{d(x,y)}{d(x,\alpha)+d(y,\alpha)+d(x,y)}$$
Consequently, the probabilistic Jaccard loss distance $\mathcal{J}_{c}$ defined in \eqref{eq:jaccard_prob_2}:
\begin{equation}
  \forall u,v \in [0,1]^{C\times N}, \ \mathcal{J}_{c}(u,v) = \frac{2\norm{u_c-v_c}_{1}}{\norm{u_c}_{1}+\norm{v_c}_{1}+\norm{u_c-v_c}_{1}}
\end{equation}
can be seen as a Steinhaus transform of the metric space $([0,1]^{N}, ||.||_{1})$ with $\alpha=0$ and thus is a distance.
Given that the weighted sum of distances is a distance, we finally conclude that the probabilistic multi-class Jaccard defined as:
\begin{equation}
   \mathcal{J} = \sum_{c\in C} \mathcal{J}_{c}
\end{equation}
is a distance.

%%%

\section{Proof of Proposition 1}\label{appendix:prop}
First, Equations \eqref{eq:1}, \eqref{eq:scenario1_final} and \eqref{eq:upbnd_init} are combined:

\begin{equation}\label{eq:appendix_initial_DA}
\begin{split}
    \mathbb{E}_{\mathcal{D}_{lesion}}\left[\mathcal{L}\left(h_{\theta}(x), y\right)\right]  &\leq \mathcal{R}_{seg}+ \epsilon_{lesion}(\theta)-\epsilon_{control}(\theta) \\
    & \leq \mathcal{R}_{seg}+ |\epsilon_{lesion}(\theta)-\epsilon_{control}(\theta)|
\end{split}
\end{equation}
where $\epsilon_{lesion}(\theta)$ and $\epsilon_{control}(\theta)$ denote the expected tissue loss on the \emph{lesion} and \emph{control} domains, defined as:
\begin{equation}
\begin{split}
\epsilon_{lesion}(\theta) = \mathbb{E}_{\mathcal{D}_{lesion}}[\mathcal{L}^{T}(h_{\theta}(x^{T_1}),y^T))] \\
\epsilon_{control}(\theta) = \mathbb{E}_{\mathcal{D}_{control}}[\mathcal{L}^{T}(h_{\theta}(x^{T_1}),y^T))] 
\end{split}
\end{equation}

Let $\theta^{*}= \argmin_{\theta \in \Theta} \epsilon_{lesion}(\theta) + \epsilon_{control}(\theta)$ be the parameters of the ideal (and unknown) segmenter that minimises the two expected tissue losses. Then let denote $\epsilon_{lesion}(\theta,\theta^{*})$ and $\epsilon_{control}(\theta,\theta^{*})$ the performance gap between the segmenter parametrised by $\theta$ and this ideal segmenter:
\begin{equation}
\begin{split}
\epsilon_{lesion}(\theta,\theta^{*}) = \mathbb{E}_{\mathcal{D}_{lesion}}[\mathcal{L}^{T}(h_{\theta}(x^{T_1}),h_{\theta^{*}}(x^{T_1}))] \\
\epsilon_{control}(\theta, \theta^{*}) = \mathbb{E}_{\mathcal{D}_{control}}[\mathcal{L}^{T}(h_{\theta}(x^{T_1}),h_{\theta^{*}}(x^{T_1}))]
\end{split}
\end{equation}

Using the fact that the loss function satisfies the triangle inequality:
\begin{equation}
\begin{split}
\epsilon_{lesion}(\theta) \leq \epsilon_{lesion}(\theta,\theta^{*}) + \epsilon_{lesion}(\theta^{*}) \\
\epsilon_{control}(\theta) \leq \epsilon_{control}(\theta,\theta^{*}) + \epsilon_{control}(\theta^{*})
\end{split}
\end{equation}

Then, the performance gap between the two domains $|\epsilon_{lesion}(\theta)-\epsilon_{control}(\theta)|$ can be bounded as follows:
\begin{equation}\label{eq:appendix_bendavid}
    |\epsilon_{lesion}(\theta)-\epsilon_{control}(\theta)| \leq |\epsilon_{lesion}(\theta,\theta^{*}) - \epsilon_{control}(\theta,\theta^{*})| + \epsilon(\Theta)
\end{equation}
where $\epsilon(\Theta)$  is the tissue expected loss of the ideal segmenter:

$$\epsilon(\Theta)=\epsilon_{lesion}(\theta^*) + \epsilon_{control}(\theta^*)$$

\eqref{eq:appendix_bendavid} can be found in \cite{Ben-David2010} in which the loss function is assumed to be the $L_1$ distance.

Given that $\epsilon(\Theta)$ is a constant w.r.t the network parameters $\theta$, the goal of domain adaptation is to reduce the distribution discrepancy $d_{DA}(\theta)$ defined as:
\begin{equation}\label{eq:d_DA_appendix}
    d_{DA}(\theta) = |\epsilon_{lesion}(\theta,\theta^{*}) - \epsilon_{control}(\theta,\theta^{*})|
\end{equation}

Similarly to \cite{NIPS2018_7436}, we demonstrate that this discrepancy $d_{DA}(\theta)$ can be estimated using the discriminator accuracy. 

Let denote $\mathcal{D}_{l}^{\theta}=(x^{T_1}, f_{\theta}(x^{T_1}))_{x^{T_1} \sim \mathcal{D}_{lesion}}$ and  $\mathcal{D}_{c}^{\theta}=(x^{T_1}, f_{\theta}(x^{T_1}))_{x^{T_1} \sim \mathcal{D}_{control}}$  the proxies of the distributions $\mathcal{D}_{lesion}$ and $\mathcal{D}_{control}$. Then, the two performance gaps with the ideal segmenter can be re-written as:
\begin{equation}\label{eq:proxies}
    \begin{split}
        \epsilon_{lesion}(\theta, \theta^{*}) = \mathbb{E}_{(x,f)\sim\mathcal{D}^{\theta}_{l}}[\mathcal{L}^{T}(h_{\theta^{*}}(x^{T_1}),f)] \\
        \epsilon_{lesion}(\theta, \theta^{*}) = \mathbb{E}_{(x,f)\sim\mathcal{D}^{\theta}_{l}}[\mathcal{L}^{T}(h_{\theta^{*}}(x^{T_1}),f)]
    \end{split}
\end{equation}
Let also define a difference hypothesis space $\Delta$:
$$\Delta \triangleq \{ \delta_{\theta^{'}} \colon (x^{T_1},f) \mapsto \mathcal{L}^{T}(h_{\theta^{'}}(x^{T_1}),f)), \ \theta^{'}\in\Theta\}$$
Moreover, we define the $\Delta$-distance between the two distributions $\mathcal{D}_{l}^{\theta}$ and $\mathcal{D}_{c}^{\theta}$ as:
\begin{equation}\label{eq:delta_distance}
    d_\Delta(\mathcal{D}_{l}^{\theta},\mathcal{D}_{c}^{\theta}) \triangleq \sup_{\delta_{\theta^{'}} \in \Delta} \left|\mathbb{E}_{\mathcal{D}_{l}^{\theta}}[\delta_{\theta^{'}} (x^{T_1},f)]-\mathbb{E}_{\mathcal{D}_{c}^{\theta}}[\delta_{\theta^{'}} (x^{T_1},f)]\right|
\end{equation}

Finally, by combining \eqref{eq:d_DA_appendix}, \eqref{eq:proxies} and \eqref{eq:delta_distance}, we obtain the following upper bound for the distribution discrepancy $d_{DA}(\theta)$:
\begin{equation}\label{eq:appendix_delta}
    \begin{split}
        d_\Delta(\mathcal{D}_{l}^{\theta},\mathcal{D}_{c}^{\theta}) &= \sup_{\delta_{\theta^{'}} \in \Delta} \left|\mathbb{E}_{\mathcal{D}_{l}^{\theta}}[\delta_{\theta^{'}} (x^{T_1},f)]-\mathbb{E}_{\mathcal{D}_{c}^{\theta}}[\delta_{\theta^{'}} (x^{T_1},f)]\right| \\
        &= \sup_{\theta^{'}\in\Theta} \left|\mathbb{E}_{\mathcal{D}^{\theta}_{l}}[\mathcal{L}^{T}(h_{\theta^{'}}(x^{T_1}),f)] - \mathbb{E}_{\mathcal{D}^{\theta}_{c}}[\mathcal{L}^{T}(h_{\theta^{'}}(x^{T_1}),f)]\right| \\
        &\geq  \underbrace{|\mathbb{E}_{\mathcal{D}^{\theta}_{l}}[\mathcal{L}^{T}(h_{\theta^{*}}(x^{T_1}),f)] - \mathbb{E}_{\mathcal{D}^{\theta}_{c}}[\mathcal{L}^{T}(h_{\theta^{*}}(x^{T_1}),f)]|}_{ =d_{DA}(\theta)}
    \end{split}
\end{equation}

Finally, let $K>0$ be the upper bound of the loss function $\mathcal{L}^{T}$ and $\mathcal{H}^{K}_{\Phi}$ denote the family of the discriminators multiplied by $K$:
$$\mathcal{H}^{K}_{\Phi} \triangleq \{K D_{\phi}, \ \phi\in\Phi\}$$
$\mathcal{H}^{K}_{\Phi}$ and the difference hypothesis space $\Delta$ are two continuous and $K$-bounded function classes.
Similarly to \cite{NIPS2018_7436}, let's assume that the family of the discriminator $\mathcal{H}^{K}_{\Phi}$ is rich enough to contain the difference hypothesis space $\Delta$. Given that a multilayer perceptrons that can fit any functions, this assumption is not unrealistic. Then, we show that the discriminator accuracy of the best discriminator is an upper bound of the $\Delta$-distance:
\begin{equation}\label{eq:appendix_dicriminator}
    \begin{split}
        d_\Delta(\mathcal{D}_{l}^{\theta},\mathcal{D}_{c}^{\theta}) &\leq K\sup_{\phi} \left|\mathbb{E}_{\mathcal{D}_{l}^{\theta}}[ D(x^{T_1},f)]-\mathbb{E}_{\mathcal{D}_{c}^{\theta}}[ D(x^{T_1},f)]\right| \\
        &\underbrace{\leq K\sup_{{\phi}} \left|\mathbb{E}_{\mathcal{D}_{l}^{\theta}}[ D(x^{T_1},f)]+\mathbb{E}_{\mathcal{D}_{c}^{\theta}}[ 1-D(x^{T_1},f)]\right|}_{=K\sup_{{\phi}} \mathcal{R}_{DA}(\phi,\theta)} 
    \end{split}
\end{equation}

Finally, by combining \eqref{eq:appendix_initial_DA}, \eqref{eq:appendix_bendavid}, \eqref{eq:appendix_delta} \eqref{eq:appendix_dicriminator}, we obtain the following tractable upper bound for the expect segmentation loss:
\begin{equation}\label{eq:appendix_final_DA}
 \mathbb{E}_{\mathcal{D}_{lesion}}\left[\mathcal{L}\left(h_{\theta}(x), y\right)\right]  \leq \mathcal{R}_{seg}+ K\sup_{\phi}\mathcal{R}_{DA}(\phi,\theta)  + \epsilon(\Theta) 
\end{equation}

\section{Visualisation symmetrised brain scans}
Figure \ref{fig:symetrized_brains} shows some examples of pseudo-healthy scans, with their tissue annotations, synthesised as described in \ref{sec:pseudo_healthy}.
\begin{figure*}[tbhp!]
 % Caption and label go in the first argument and the figure contents
 % go in the second argument
  \centering
  \includegraphics[width=\linewidth]{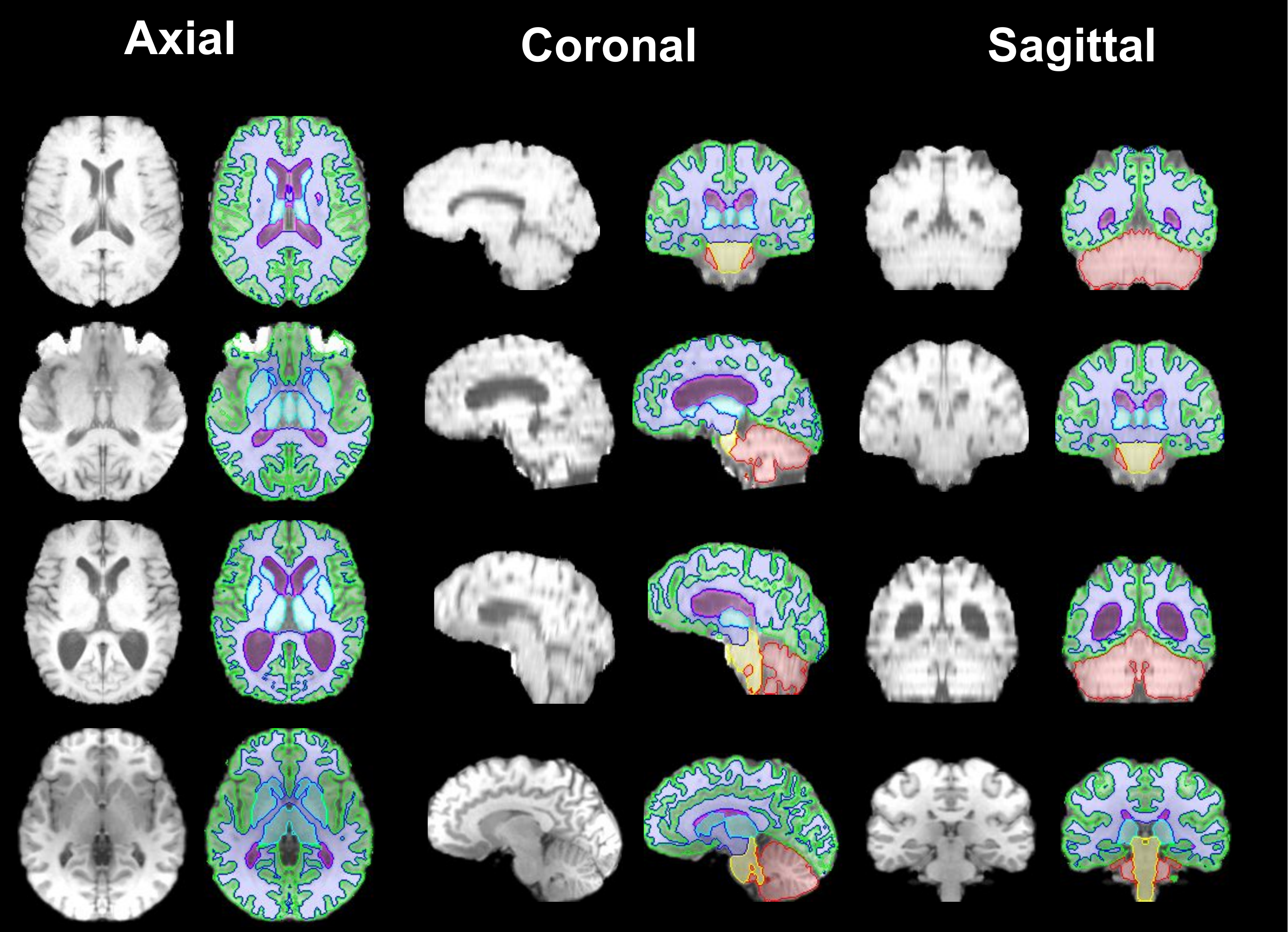}
  \caption{In order to synthesise a pseudo-healthy set of scans, we symmetrized the "healthy" hemisphere of brains from BraTS. GIF framework is then used to generate tissue ground truth. \Tone scans are shown with the tissue segmentation}
  \label{fig:symetrized_brains}
\end{figure*}

\end{document}